\documentclass[english]{article}
\usepackage[T1]{fontenc}
\usepackage[latin1]{inputenc}
\usepackage{geometry}
\usepackage{float}
\usepackage{mathrsfs}
\usepackage{amsmath}
\usepackage{amssymb}
\usepackage{graphicx}
\usepackage{subfigure}

\makeatletter

\floatstyle{ruled}
\newfloat{algorithm}{tbp}{loa}
\providecommand{\algorithmname}{Algorithm}
\floatname{algorithm}{\protect\algorithmname}

\usepackage{amsthm}

\usepackage{mathrsfs}

\usepackage{amsfonts}

\usepackage{epsfig}

\usepackage{bm}

\usepackage{mathrsfs}

\usepackage{enumerate}

\@ifundefined{definecolor}{\@ifundefined{definecolor}
 {\@ifundefined{definecolor}
 {\usepackage{color}}{}
}{}
}{}

\usepackage{subfig}\usepackage[all]{xy}

\newcommand{\bbR}{\mathbb R}

\newtheorem{theorem}{Theorem}[section]
\newtheorem{lem}{Lemma}[section]
\newtheorem{rem}{Remark}[section]
\newtheorem{prop}{Proposition}[section]

\newcounter{hypA}
\newenvironment{hypA}{\refstepcounter{hypA}\begin{itemize}
  \item[({\bf A\arabic{hypA}})]}{\end{itemize}}
\newcounter{hypB}

\newcounter{hypC}

\usepackage{babel}\date{}

\usepackage{babel}

\makeatother

\usepackage{babel}

\newcommand{\bbE}{\mathbb{E}}

\newcommand{\bbP}{\mathbb{P}}

\usepackage{xspace}
\usepackage{tabu}
\usepackage{booktabs}

\def\bbR{\mathbb{R}}

\begin{document}

\begin{center}

{\Large \textbf{Unbiased Estimation of the Gradient of the Log-Likelihood in Inverse Problems}}

\vspace{0.5cm}

BY AJAY JASRA$^{1}$, KODY J. H. LAW$^{2}$ \& DENG LU$^{3}$

{\footnotesize $^{1}$Computer, Electrical and Mathematical Sciences and Engineering Division, King Abdullah University of Science and Technology, Thuwal, 23955, KSA.}
{\footnotesize E-Mail:\,} \texttt{\emph{\footnotesize ajay.jasra@kaust.edu.sa}}\\
{\footnotesize $^{2}$Department of Mathematics,
University of Manchester, Manchester, M13 9PL, UK.}
{\footnotesize E-Mail:\,} \texttt{\emph{\footnotesize kodylaw@gmail.com}}\\
{\footnotesize $^{3}$Department of Statistics \& Applied Probability,
National University of Singapore, Singapore, 117546, SG.}
{\footnotesize E-Mail:\,} \texttt{\emph{\footnotesize denglu@u.nus.edu}}

\begin{abstract}
We consider the problem of estimating a parameter 
$\theta\in\Theta\subseteq\mathbb{R}^{d_{\theta}}$ associated
to a Bayesian inverse problem. 
Treating the unknown initial condition as a nuisance parameter, 
typically one must resort to a numerical approximation
of gradient of the log-likelihood and 
also adopt a discretization of the problem in space and/or time. 
We develop a new methodology to unbiasedly
estimate the gradient of the log-likelihood with respect to 
the unknown parameter, i.e.
the expectation of the estimate has no 
discretization bias. 
Such a property is not only useful for 
estimation in terms of the original stochastic model of interest, 
but can be used in stochastic gradient algorithms which
benefit from unbiased estimates. 
Under appropriate assumptions, we prove 
that our estimator is not only unbiased but of finite variance. 
In addition, 
when implemented on a single processor, we show that the cost to achieve a given level of error is comparable to multilevel Monte Carlo methods, both 
practically and theoretically. 
However, the new algorithm provides the possibility for parallel computation
on arbitrarily many processors without any loss of efficiency, asymptotically. 
In practice, this means any precision can be achieved in a fixed, 
finite constant time, provided that enough processors are available.

\vspace{5pt}
\end{abstract}

\end{center}

\section{Introduction}


The problem of inferring unknown parameters associated to the solution of (partial) 
differential equations (PDEs) is referred to as an inverse problem. 
In such a context, when the forward problem is well-posed, 
the inverse problem is often ill-posed and challenging to solve, even numerically. 
The area has a long history and a large 
literature (see e.g.~\cite{engl,tikhonov}) 
yet the intersection with statistics is still comparatively small,
particularly considering the significant intersection, in terms of both methods and 
algorithms as well as objectives. If one adopts a Bayesian approach to solution
of the inverse problem then the object of interest is a posterior distribution, and in particular 
expectations with respect to this distribution \cite{franklin, stuart}.
While this provides an elegant solution and quantified uncertainty via well-defined target
distribution, it is more challenging to solve than its deterministic counterpart,
requiring at least a Hessian in addition to a maximum a posteriori estimator for a Laplace approximation, 
if not more expensive Monte Carlo methods. 
Here we assume solution of the Bayesian inverse problem (BIP) requires 
computationally intensive Monte Carlo methods for accurate estimation.
 We furthermore assume that the statistical model can only be defined up to some unknown parameters.
 
Consider a BIP with unknown $u\in \mathsf{X}$ 
and data $y \in \mathsf{Y}$, 
related through a PDE, and assume that the statistical model is known 
only up to some parameter $\theta \in \Theta\subseteq\mathbb{R}^{d_{\theta}}$ 
(assumed finite dimensional). 
In other words, the posterior distribution takes the form
$$
p(du, \theta | y) \propto p(y | u, \theta) p(du | \theta) p(\theta) \, .
$$
Due to sensitivity with respect to the parameter $\theta$
and strong correlation with the unknown $u$, such posterior distribution can be 
highly complex 
and very challenging to sample from, even using quite
advanced Markov chain Monte Carlo (MCMC) algorithms. 
In this article, the unknown $u$ 
is treated as a nuisance parameter and the goal is to maximize
the marginal likelihood of the parameters 
$$
p(y | \theta) = \int_\mathsf{X} p(y| u, \theta) p(du|\theta) \, .
$$ 
In such a scenario one is 
left with a finite-dimensional optimization problem, 
albeit with an objective function that is not available analytically. 
This intractability arises 
from two sources:
\begin{itemize}
\item first, for a given $(u,\theta)$ only a discretization of the likelihood 
$p(y| u, \theta)$ can be evaluated;
\item second, 
the discretized marginal likelihood 
is a high-dimensional integral which itself must be approximated. 
\end{itemize}
Moreover, the associated gradient of the log-likelihood is not available, 
which may be of interest in optimization algorithms. 
In the following we will suppress the notation for fixed observation $y$
and present the method generally. In particular, we use the notation
$\gamma_\theta(u) = p(y | u, \theta) p(u | \theta)$,
where $du$ represents the finite measure of an infinitesimal volume element, 
which may or may not be Lebesgue measure, and $p(u | \theta) = (p(du|\theta)/du)(u)$.
We will also denote its integral $Z_\theta = p(y | \theta)$, and the posterior 
by $\eta_\theta(du)$.

In this article, we present a new scheme to provide finite variance estimates of the gradient of the log-likelihood that are unbiased. 
To be precise, let $E_{\theta} = \nabla_\theta \log ( Z_\theta )$ 
denote the gradient of the log-likelihood with no discretization bias.
The proposed method provides an estimator $\hat{E}_{\theta}$
such that $\mathbb{E}[\hat{E}_{\theta}]=E_{\theta}$, 
where $\mathbb{E}$ is the expectation with respect to the randomization induced by our numerical approach. Moreover, the estimator  $\hat{E}_{\theta}$ is constructed so that one only needs access to finite resolution (discretized) approximations of the BIP. 
This scheme 
is of interest for several reasons:
\begin{enumerate}
\item{Unbiased estimates of gradients help to facilitate stochastic gradient algorithms};
\item{The method is easy to parallelize};
\item{The method helps to provide a benchmark for other computations.}
\end{enumerate}
In terms of the first point, it is often simpler to verify the validity of stochastic gradient algorithms when the estimate of the noisy functional is unbiased. 
Whilst this is not always needed (see \cite{tadic} for a special case, which does not apply in our context), it at least provides the user a peace-of-mind when implementing optimization schemes. 
The second point is of interest, in terms of efficiency of application, especially relative to competing methods. 
The third point simply states that 
one can check the precision of biased methodology. 
We now explain the approach in a little more detail.

The method that we use is based upon a technique developed in \cite{ubpf}. 
In that article the authors consider the filtering of a class of diffusion processes, which have to be discretized. 
The authors develop a method which allows one to approximate the filtering distribution, unbiasedly and without any discretization error. 
The methodology that is used
in \cite{ubpf} is a double randomization scheme based upon the approaches in \cite{mcl,rhee}. 
The work in \cite{mcl,rhee} provides a methodology to turn a sequence of convergent
estimators into an unbiased estimator, 
using judicious randomization across the level of discretization. 
It is determined for the problem of interest in \cite{ubpf} 
that an additional randomization is required in order to derive efficient estimators, that is, estimators that are competitive with the existing state-of-the-art methods in the literature. 
In this article we follow the basic approach that is used in \cite{ubpf}, 
except that one cannot use the same estimation methodology for the current problem.
An approach is introduced in \cite{beskos} which enables 
application of the related deterministic 
multilevel Monte Carlo 
identity \cite{vihola} to a sequential Monte Carlo (SMC) 
sampler \cite{delm:13, delm:04} for inference in the present context. 
In this article, we consider such a strategy
to allow the application of the approach in \cite{ubpf} to unbiasedly estimate the gradient of the log-likelihood for BIPs. 
The method of \cite{beskos} is 
one of the most efficient techniques that could be used
for estimation of the gradient of the log-likelihood for BIPs. 
However, this method is subject to discretization bias. 
In other words, suppose
$E_{\theta}^l$ is the gradient of the log-likelihood with 
a choice of discretization bias level, e.g. $2^{-l}$.
The original method would produce an estimate $\hat{E}_{\theta}^l$ for which
$\mathbb{E}[\hat{E}_{\theta}^l] \neq E_{\theta}^l$. 
On the other hand, under assumptions, it is proven that
the new method introduced here can produce an estimate 
$E_{\theta}$ with finite variance and without bias, i.e. 
$\mathbb{E}[\hat{E}_{\theta}] = E_{\theta}^\infty$. 
We also show that the cost to achieve a given variance is very similar to the multilevel SMC (MLSMC) approach of \cite{beskos}, with high probability. 
This is confirmed in numerical simulations. 
We furthermore numerically investigate 
the utility of 
our new estimator in the context of stochastic gradient algorithms,
where it is shown that a huge improvement in efficiency is possible.
Our approach is one of the first which can 
in general provide unbiased and finite variance estimators of the gradient of the log-likelihood for BIPs. 
A possible alternative would be the approach of \cite{sergios}, however, 
the methodology in that article is not as general as is presented here and may 
be more challenging to implement.

This article is structured as follows. In Section \ref{sec:problem} we explain the generic problem to which 
our approach is applicable. 
In particular, a concrete
example in the context of Bayesian inverse problems is described. In Section \ref{sec:method} we present our methodology and the proposed estimator. 
In Section \ref{sec:theory} we show that our proposed estimator is unbiased and of finite variance and we consider 
the cost to obtain the estimate. 
In Section \ref{sec:numerics} several numerical
examples are presented to investigate performance of the estimator in practice, 
including the efficiency of the estimator when used in 
in the relevant context of a stochastic gradient algorithm for parameter estimation. 
In appendix \ref{app:proofs} the proofs
of some of our theoretical results can be found.

\section{Problem Setting}\label{sec:problem}

\subsection{Generic Problem}

Let $(\mathsf{X},\mathcal{X})$ be a measurable space, 
and define a probability measure on it as
$$
\eta_{\theta}(du) = \frac{\gamma_{\theta}(u)du}{\int_{\mathsf{X}}\gamma_{\theta}(u)du}
$$
where $\theta\in\Theta\subseteq\mathbb{R}^{d_{\theta}}$, $\gamma:\Theta\times\mathsf{X}\rightarrow\mathbb{R}_+$ and $du$ is a
$\sigma-$finite measure on $(\mathsf{X},\mathcal{X})$. We are interested in computing
\begin{eqnarray*}
\nabla_{\theta}\log\Big(\int_{\mathsf{X}}\gamma_{\theta}(u)du\Big) & = & 
\int_{\mathsf{X}}\nabla_{\theta}\log\Big(\gamma_{\theta}(u)\Big) \eta_{\theta}(du) \\
& = & \int_{\mathsf{X}}\varphi_{\theta}(u) \eta_{\theta}(du) \, ,
\nonumber
\end{eqnarray*}
where we have defined 
$\varphi_{\theta}(u) = \nabla_{\theta}\log\Big(\gamma_{\theta}(u)\Big)$.
From here on, we will use the following short-hand notation for a measure $\mu$ on $(\mathsf{X},\mathcal{X})$ and a measurable $\mu-$integrable
$\varphi:\mathsf{X}\rightarrow\mathbb{R}^d$
$$\mu(\varphi):=\int_{\mathsf{X}}\varphi(x)\mu(dx) \, ,$$ which should be understood as a column vector of integrals.

In practice, we assume that we must work with an approximation of $\varphi_{\theta}(u)$ and $\eta_{\theta}(du)$. Let $l\in\mathbb{N}_0$, and set 
$$
\eta_{\theta}^l(du) = \frac{\gamma_{\theta}^l(u)du}{\int_{\mathsf{X}}\gamma_{\theta}^l(u)du}
$$
where $\gamma^l:\Theta\times\mathsf{X}\rightarrow\mathbb{R}_+$. We are now interested in computing
\begin{eqnarray*}
\nabla_{\theta}\log\Big(\int_{\mathsf{X}}\gamma_{\theta}^l(u)du\Big) & = & \int_{\mathsf{X}}\nabla_{\theta}\log\Big(\gamma_{\theta}^l(u)\Big) \eta_{\theta}^l(du) \\
& = & \int_{\mathsf{X}}\varphi_{\theta}^l(u) \eta_{\theta}^l(du).
\nonumber
\end{eqnarray*}
It is assumed explicitly that $\forall \theta\in\Theta$
$$
\lim_{l\rightarrow+\infty}\eta_{\theta}^l(\varphi_{\theta}^l) = \eta_{\theta}(\varphi_{\theta}).
$$

\subsection{Example of Problem}\label{sec:example}

We will focus on the following particular problem.
Let $D\subset\mathbb{R}^d$ with $\partial D\in C^1$ convex
and $f\in L^2(D)$. 
Consider the following PDE on $D$:

\begin{align}
-\nabla \cdot (\hat{u}\nabla p)  &=f,\quad \textrm{ on } D,       \\
p&= 0, \quad \textrm{ on } \partial D,
\nonumber 
\end{align}
where
$$
\hat{u}(x) = \bar{u}(x) + \sum_{k=1}^Ku_k\sigma_k\phi_k(x).
$$
Define $u=\{u_k\}_{k=1}^K$, with $u_k \sim U[-1,1]$ i.i.d. 
(the uniform distribution on $[-1,1]$). 
This determines the prior distribution for $u$. The state space is 
$\mathsf{X}=\prod_{k=1}^K[-1,1]$. Let $p(\cdot;u)$ denote the weak solution of $(1)$ for parameter value $u$. The following will be assumed.

\begin{hypA}
\label{hyp:N}
$f, \phi_k \in C(D)$, $\|\phi_k\|_\infty \leq 1$, 
and there is a $u_*>0$ such that 
$\bar{u}(x) > \sum_{k=1}^K\sigma_k + u_*$. 
\end{hypA}

Note that this assumption guarantees $\hat u > u_*$
uniformly in $u$, hence there is a 
well-defined (weak) solution $p(\cdot ;u)$
which will be bounded in uniformly in $u$, 
in an appropriate space, e.g. $L^2(D)$.

Define the following vector-valued function 
$$
\mathcal{G}(u) = [g_1(p(\cdot;u)),\dots,g_M(p(\cdot;u))]^{\intercal},
$$
where $g_m$ are elements of the dual space
(e.g. $L^2(D)$ is sufficient), for $m=1,\dots,M$. 
It is assumed that the data take the form 
$$
y = \mathcal{G}(u) + \xi,\quad \xi \sim N(0,\theta^{-1}\cdot\bm{I}_M),\quad \xi 
\perp u,
$$
where $N(0,\theta^{-1}\cdot\bm{I}_M)$ denotes the Gaussian random variable with mean $0$ and covariance matrix $\theta^{-1}\cdot\bm{I}_M$, and $\perp$ 
denotes independence. 
The unnormalized density of $u$ for fixed $\theta$ is then given by
\begin{equation}\label{eq:unno}
\gamma_{\theta}(u) = \theta^{M/2}\exp(-\frac{\theta}{2}\|\mathcal{G}(u) - y\|^2) \, ,
\end{equation}
the normalized density is given by
$$
\eta_{\theta}(u) = \frac{\gamma_{\theta}(u)} {Z_\theta} \, ,
$$
where $Z_\theta = {\int_{\mathsf X}\gamma_{\theta}(u)du}$, 
and the quantity of interest is defined as 
\begin{equation}\label{eq:phi}
\varphi_\theta(u) := \nabla_{\theta}\log\Big(\gamma_{\theta}(u)\Big) 
= \frac{M}{2\theta} - \frac{1}{2}\|\mathcal{G}(u) - y\|^2) \, .
\end{equation}

\subsubsection{Particular setup}\label{ssec:example}

Let $d=1$ and $D=[0,1]$ and consider $f(x)=100x$. For the prior specification of $u$, we set $K=2, \bar{u}(x)=0.15$, and for $k>0$, let $\sigma_k=(2/5)4^{-k}, \phi_k(x)=\sin(k\pi x)$ if $k$ is odd and $\phi_k(x)=\cos(k\pi x)$ if $k$ is even. The observation operator is $\mathcal{G}(u)=[p(0.25;u),p(0.75;u)]^{\intercal}$, and the parameter in observation noise covariance is taken to be $\theta=0.3$.

The PDE problem at resolution level $l$ is solved using a finite element method with piecewise linear shape functions on a uniform mesh of width $h_l=2^{-l}$, for $l\geq2$. Thus, on the $l$th level the finite-element basis functions are $\{\psi_i^l\}_{i=1}^{2^l-1}$ defined as (for $x_i = i\cdot 2^{-l}$):
$$
\psi_i^l(x) =  \left\{\begin{array}{ll}
(1/h_l)[x-(x_i-h_l)]
& \textrm{if}~x\in[x_i-h_l,x_i], \\
(1/h_l)[x_i+h_l-x] & \textrm{if}~x\in[xi,x_i+h_l].
\end{array}\right.
$$
To solve the PDE, $p^l(x)=\sum_{i=1}^{2^l-1}p_i^l\psi_i^l(x)$ is plugged into (1), and projected onto each basis element:
$$
-\Big\langle \nabla\cdot\Big(\hat{u}\nabla\sum_{i=1}^{2^l-1}p_i^l \psi_i^l \Big),\psi_j^l \Big\rangle = \langle f, \psi_j^l \rangle, 
$$
resulting in the following linear system:
$$
\bm{A}^l(u)\bm{p}^l = \bm{f}^l,
$$
where we introduce the matrix $\bm{A}^l(u)$ with entries $A_{ij}^l(u) = \langle \hat{u}\nabla\psi_i^l,\nabla\psi_j^l \rangle$, and vectors $\bm{p}^l, \bm{f}^l$ with entries $p_i^l$ and $f_i^l=\langle f, \psi_i^l\rangle$, respectively.

Define $\mathcal{G}^l(u) = [g_1(p^l(\cdot;u)),\dots,g_M(p^l(\cdot;u))]^{\intercal}$.
Denote the corresponding approximated unnormalized density by 
\begin{equation}\label{eq:unnol}
\gamma_{\theta}^l(u) = \theta^{M/2}\exp(-\frac{\theta}{2}\|\mathcal{G}^l(u) - y\|^2),
\end{equation}
and the approximated normalized density by
$$
\eta_{\theta}^l(u) = \frac{\gamma_{\theta}^l(u)} {Z_\theta^l} \, ,
$$
where $Z_\theta^l = {\int_{E}\gamma_{\theta}^l(u)du}$.
We further define 
\begin{equation}\label{eq:phil}
\varphi^l_\theta(u) := \nabla_{\theta}\log\Big(\gamma_{\theta}^l(u)\Big) 
= \frac{M}{2\theta} - \frac{1}{2}\|\mathcal{G}^l(u) - y\|^2) \, .
\end{equation}

It is well-known that under assumption (A\ref{hyp:N}) 
$p^l$ converges to $p$ as $l \rightarrow \infty$ uniformly in $u$ \cite{brenner, ciarlet}.
Furthermore, continuity ensures $\gamma_{\theta}^l(u)$ 
converges to $\gamma_{\theta}(u)$ and  
$\varphi^l_\theta(u)$ converges to $\varphi_\theta(u)$ uniformly in $u$ as well. 

\section{Methodology for Unbiased Estimation}\label{sec:method}

We now describe our methodology for computing an unbiased estimate of $\eta_{\theta}(\varphi_{\theta})$. 
For simplicity of exposition we will suppose that for $i\in\{1,\dots,d_{\theta}\}$, $(\varphi_\theta(u))_i\in\mathcal{B}_b(\mathsf{X})$, where
$(x)_i$ denotes the $i^{th}-$element of a vector and $\mathcal{B}_b(\mathsf{X})$ are the collection of bounded, measurable and real-valued functions on $\mathsf{X}$.
This constraint is not needed for the numerical implementation of the method, but, shall reduce most of the technical exposition to follow.
As remarked in the introduction, the basic approach
follows that in \cite{ubpf} with some notable differences. We now detail how the approach will work.

\subsection{Methodology in \cite{ubpf}}

The underlying approach of \cite{ubpf} is a type of double randomization scheme. The first step is to use the single-term estimator
as developed in \cite{rhee}. Suppose one wants to estimate $\eta_{\theta}(\varphi_{\theta})$, but, only has access to a methodology
that can approximate $\eta_{\theta}^l(\varphi_{\theta}^l)$ for each fixed $l\in\mathbb{N}_0$. Let $\mathbb{P}_L(l)$ be a positive probability
mass function on $\mathbb{N}_0$ and suppose that one can construct a sequence of random variables $(\Xi_{\theta}^l)_{l\geq 0}$ 
such that 
\begin{eqnarray}
\mathbb{E}[\Xi_{\theta}^0] & = & \eta_{\theta}^0(\varphi_{\theta}^0) \label{eq:ub1}\\
\mathbb{E}[\Xi_{\theta}^l] & = & \eta_{\theta}^l(\varphi_{\theta}^l) - \eta_{\theta}^{l-1}(\varphi_{\theta}^{l-1})\quad\quad l\in\mathbb{N}\label{eq:ub2}
\end{eqnarray}
and that
\begin{equation}\label{eq:ub3}
\sum_{l\in\mathbb{N}_0}\frac{1}{\mathbb{P}_L(l)}\mathbb{E}[\|\Xi_{\theta}^l\|^2] < +\infty
\end{equation}
where $\|\cdot\|$ is the $L_2-$norm. 
Now if one draws $L\sim\mathbb{P}_L(\cdot)$, then 
$\Xi_{\theta}^L/\mathbb{P}_{L}(L)$ is an unbiased and finite variance estimator of $\eta_{\theta}(\varphi_{\theta})$. It should be noted that \eqref{eq:ub1}-\eqref{eq:ub2}
are not necessary conditions, but are sufficient to ensure the unbiasedness of the estimator. 

In the context of interest, it can be challenging to obtain a sequence of random variables which can possess the properties \eqref{eq:ub1}-\eqref{eq:ub3}. We will
detail one possible approach at a high-level and then explain in details how one can actually construct a simulation method to achieve this high-level description.

\subsection{High Level Approach}
\label{ssec:high}

The objective of this section is to highlight the generic procedure that is used in \cite{ubpf} for producing estimates that 
satisfy \eqref{eq:ub1}-\eqref{eq:ub2}.
The basic idea is to use another application of randomization 
to construct such unbiased estimators from a consistent sequence of estimators.
In particular, consider a given increasing sequence $(N_p)_{p\in\mathbb{N}_0}$ with $N_p\in\mathbb{N}$ for each $p\in\mathbb{N}_0$, $1\leq N_0< N_1< \cdots$ and
$\lim_{p\rightarrow\infty}N_p=\infty$. Then, we suppose that one can construct 
$N_p$-sample Monte Carlo (type) estimators 
$\xi_\theta^{l,p}$ for $l\in\mathbb{N}_0$,
such that almost surely the following consistency results hold
\begin{eqnarray}\label{eq:cons1}
\lim_{p\rightarrow\infty}\xi_\theta^{0,p}
& = & \eta_{\theta}^{0}(\varphi_{\theta}^0) \, ,\\
\lim_{p\rightarrow\infty} \xi_\theta^{l,p}
& = &\eta_{\theta}^{l}(\varphi_{\theta}^l)-\eta_{\theta}^{l-1}(\varphi_{\theta}^{l-1})
\, , \quad\quad l\in\mathbb{N} \, .
\label{eq:cons2}
\end{eqnarray}
For a given $(l,p,p')\in\mathbb{N}_0^3$, $p\neq p'$ we do \emph{not} require 
$\xi_\theta^{l,p}$ and $\xi_\theta^{l,p'}$
to be independent, nor do we require unbiasedness of the individual estimators as in
\begin{eqnarray*}
\mathbb{E}[\xi_\theta^{0,p}]
& = & \eta_{\theta}^{0}(\varphi_{\theta}^0) \, , \\
\mathbb{E}[\xi_\theta^{l,p}]
& = &
\eta_{\theta}^{l}(\varphi_{\theta}^l)-\eta_{\theta}^{l-1}(\varphi_{\theta}^{l-1}) \, , \quad\quad l\in\mathbb{N} \, .
\end{eqnarray*}
Now set 
\begin{eqnarray*}
\Xi_{\theta}^{0,0} & := & \xi_\theta^{0,0} \, ,\\
\Xi_{\theta}^{0,p} & := &  \xi_\theta^{0,p} -  \xi_\theta^{0,p-1} \, , 
\quad\quad p\in\mathbb{N} \, .
\end{eqnarray*}
For $l\in\mathbb{N}$ given, set
\begin{eqnarray*}
\Xi_{\theta}^{l,0} & := & 
\xi_\theta^{l,0} \, , \\
\Xi_{\theta}^{l,p} & := & \xi_\theta^{l,p} -  \xi_\theta^{l,p-1} \, , \quad\quad p\in\mathbb{N} \, .
\end{eqnarray*}
Let $\mathbb{P}_P(p)$, $p\in\mathbb{N}_0$, be a positive probability mass function with $\overline{\mathbb{P}}_P(p)=\sum_{q=p}^{\infty}\mathbb{P}_P(q)$.
Now if
\begin{eqnarray}\label{eq:cs_fv}
\sum_{p\in\mathbb{N}_0}\frac{1}{\overline{\mathbb{P}}_P(p)}
\mathbb{E}[\| \xi_\theta^{l,p}
-\eta_{\theta}^0(\varphi_{\theta}^0)\|^2]
& < & +\infty \label{eq:cs_fv1} \, , \\
\sum_{p\in\mathbb{N}_0}\frac{1}{\overline{\mathbb{P}}_P(p)}\mathbb{E}[\|
 \xi_\theta^{l,p} 
- \{\eta_{\theta}^{l}(\varphi_{\theta}^l) - \eta_{\theta}^{l-1}(\varphi_{\theta}^{l-1})\}
\|^2]
& < & +\infty \, , \quad\quad l\in\mathbb{N}\label{eq:cs_fv2}
\end{eqnarray}
and $P\sim\mathbb{P}_P(\cdot)$, then 
\begin{equation}\label{eq:xi_l_def}
\Xi_{\theta}^l = \sum_{p=0}^P \frac{1}{\overline{\mathbb{P}}_P(p)} \Xi_{\theta}^{l,p}
\end{equation}
will allow $(\Xi_{\theta}^l)_{l\in\mathbb{N}_0}$ to satisfy \eqref{eq:ub1}-\eqref{eq:ub2},
where expectations are understood to be with respect to $\mathbb{P}_P$ 
yet $P$ is suppressed in the notation. 
Moreover $(\Xi_{\theta}^l)_{l\in\mathbb{N}_0}$ will have finite variances.
This result follows as we are simply using the coupled sum estimator as in \cite{rhee} 
and using \cite[Theorem 5]{vihola}, for instance, to verify the
conditions required. 

\subsection{Details of the Approach}

We will now describe how to obtain the sequence $(\Xi_{\theta}^{l,p})_{p\in\mathbb{N}_0}$ for $l\in\mathbb{N}_0$ fixed.

\subsubsection{MLSMC Method of \cite{beskos}}

To introduce our approach, we first consider the MLSMC method in \cite{beskos} which will form the basis for our estimation procedure.
Define for $l\in\mathbb{N}_0$
$$
G_{\theta}^l(u) = \frac{\gamma_{\theta}^{l+1}(u)}{\gamma_{\theta}^{l}(u)}
$$
and for $l\in\mathbb{N}$, $M_{\theta}^l$ is a $\eta_{\theta}^l-$invariant Markov kernel; that is, for any $\varphi\in\mathcal{B}_b(\mathsf{X})$
\begin{equation}\label{eq:m}
\eta_{\theta}^l(\varphi) = \int_{\mathsf{X}}\Big(\int_{\mathsf{X}}\varphi(u') M_{\theta}^l(u,du')\Big)\eta_{\theta}^l(du).
\end{equation}
Define for $\mu\in\mathcal{P}(\mathsf{X})$ (the collection of probability measures on $(\mathsf{X},\mathcal{X})$), $l\in\mathbb{N}$
\begin{equation}\label{eq:Phi}
\Phi_{\theta}^l(\mu)(du') := \frac{1}{\mu(G_{\theta}^{l-1})} \int_{\mathsf{X}}G_{\theta}^{l-1}(u)M_{\theta}^l(u,du') \mu(du) \, .
\end{equation}
Noting that 
\begin{equation}\label{eq:rat}
\eta_{\theta}^l(\varphi) = 
\frac{\eta_{\theta}^{l-1}(G_{\theta}^{l-1}\varphi)}{\eta_{\theta}^{l-1}(G_{\theta}^{l-1})} 
= \frac{Z_\theta^{l-1}}{Z_\theta^l}\eta_{\theta}^{l-1}(G_{\theta}^{l-1}\varphi) 
= \frac{1}{Z_\theta^l} \int_{\mathsf{X}} ( \gamma_\theta^l(u) \varphi(u) ) du
\, ,
\end{equation}
equations \eqref{eq:m} and \eqref{eq:Phi} lead to the recursion
\begin{eqnarray*}
\eta_{\theta}^l(\varphi) = 
\frac{\eta_{\theta}^{l-1}(G_{\theta}^{l-1}\varphi)}{\eta_{\theta}^{l-1}(G_{\theta}^{l-1})}
&=& \frac{1}{\eta_{\theta}^{l-1}(G_{\theta}^{l-1})} 
\int_{\mathsf{X}} G_{\theta}^{l-1}(u) 
\Big ( \int_{\mathsf{X}}\varphi(u') M_{\theta}^l(u,du') \Big)
 \eta_{\theta}^{l-1}(du)  \\
&=& \Phi_{\theta}^l(\eta_{\theta}^{l-1})(\varphi) \, .
\end{eqnarray*}
Consider $N\in\mathbb{N}$,
and slightly modify the MLSMC algorithm used in \cite{beskos} 
to keep the number of samples across layers fixed,
up to some given level $l\in\mathbb{N}$. 
Details are given in Algorithm \ref{alg:mlsmc}. 
\begin{algorithm}[h!]
\begin{enumerate}
\item{Initialization: For $i\in\{1,\dots,N\}$ sample $U_0^{i}$ from $\eta_\theta^0$. If $l=0$ stop; otherwise set $s=1$ and go-to step 2.}
\item{Resampling and Sampling: For $i\in\{1,\dots,N\}$ sample $U_s^{i}$ from $\Phi_{\theta}^s(\eta_{\theta}^{s-1,N})$. This consists of sampling $a_s^i\in\{1,\dots,N\}$ with probability
mass function
$$
\mathsf{P}_\theta^N(a_s^i=j) = \frac{G_{\theta}^{s-1}(u_{s-1}^j)}{\sum_{k=1}^N G_{\theta}^{s-1}(u_{s-1}^k)} \, ,
$$
and then sampling $U_s^i$ from $M_{\theta}^s(u_{s-1}^{a_s^i},\cdot)$. If $s=l$ stop; otherwise set $s=s+1$ and return to the start of 2.}
\end{enumerate}
\caption{A Multilevel Sequential Monte Carlo Sampler with a fixed number of samples $N\in\mathbb{N}$ and a given level $l\in\mathbb{N}_0$.}
\label{alg:mlsmc}
\end{algorithm}

This algorithm yields samples distributed according to the following joint law
\begin{equation}\label{eq:mlsmc_law}
\mathsf{P}_\theta^N\big(d(u_0^{1:N},\dots,u_l^{1:N})\big) = \Big(\prod_{i=1}^{N} \eta_{\theta}^{0}(du_0^i)\Big) 
\Big(\prod_{s=1}^l \prod_{i=1}^{N} \Phi_{\theta}^s(\eta_{\theta}^{s-1,N})(du_s^i)
\Big) \, ,
\end{equation}
where $\eta_{\theta}^{s-1,N}(du)=\frac{1}{N}\sum_{i=1}^{N}\delta_{u_{s-1}^i}(du)$ 
for $s\in\mathbb{N}$. 
One can compute an estimate of $\eta_{\theta}^0(\varphi_{\theta}^0)$ as
$$
\eta_{\theta}^{0,N}(\varphi_{\theta}^0) := 
\frac{1}{N}\sum_{i=1}^N \varphi_{\theta}^0(u_0^i) \, .
$$
Following from \eqref{eq:rat}, for $l\in\mathbb{N}$, 
one can estimate 
$\eta_{\theta}^l(\varphi_{\theta}^l)-\eta_{\theta}^{l-1}(\varphi_{\theta}^{l-1})$ with
$$
\frac{\eta_{\theta}^{l-1,N}(G_{\theta}^{l-1}\varphi_{\theta}^l)}{\eta_{\theta}^{l-1,N}(G_{\theta}^{l-1})} - \eta_{\theta}^{l-1,N}(\varphi_{\theta}^{l-1}) = 
\frac{\frac{1}{N}\sum_{i=1}^N G_{\theta}^{l-1}(u_{l-1}^i)\varphi_{\theta}^{l}(u_{l-1}^i)}{\frac{1}{N}\sum_{i=1}^N G_{\theta}^{l-1}(u_{l-1}^i)} - \frac{1}{N}\sum_{i=1}^N \varphi_{\theta}^{l-1}(u_{l-1}^i) \, .
$$
The reason for using the samples generated at level $l-1$ to estimate 
$\eta_{\theta}^l(\varphi_{\theta}^l)$ as well as 
$\eta_{\theta}^{l-1}(\varphi_{\theta}^{l-1})$ is to 
construct estimators which 
satisfying conditions such as \eqref{eq:ub3}. 
Standard results (for instance in \cite{delm:13}) allow one to prove that almost surely
\begin{eqnarray*}
\lim_{N\rightarrow\infty} \eta_{\theta}^{0,N}(\varphi_{\theta}^0) & = & \eta_{\theta}^{0}(\varphi_{\theta}^0) \\
\lim_{N\rightarrow\infty} \Big(\frac{\eta_{\theta}^{l-1,N}(G_{\theta}^{l-1}\varphi_{\theta}^l)}{\eta_{\theta}^{l-1,N}(G_{\theta}^{l-1})} - \eta_{\theta}^{l-1,N}(\varphi_{\theta}^{l-1}) \Big)& = & 
\eta_{\theta}^l(\varphi_{\theta}^l)-\eta_{\theta}^{l-1}(\varphi_{\theta}^{l-1}) \, ,\quad\quad l\in\mathbb{N} \, .
\end{eqnarray*}
Note that in general one has 
$$
\mathsf{E}_\theta^N\bigg[\Big(\frac{\eta_{\theta}^{l-1,N}(G_{\theta}^{l-1}\varphi_{\theta}^l)}{\eta_{\theta}^{l-1,N}(G_{\theta}^{l-1})} - \eta_{\theta}^{l-1,N}(\varphi_{\theta}^{l-1}) \Big)\bigg] \neq 
\eta_{\theta}^l(\varphi_{\theta}^l)-\eta_{\theta}^{l-1}(\varphi_{\theta}^{l-1}) \, ,\quad\quad l\in\mathbb{N} \, ,
$$ 
where 
$\mathsf{E}_\theta^N$
is an expectation associated to the probability in \eqref{eq:mlsmc_law}.

\subsubsection{Approach for Constructing $(\Xi_{\theta}^{l,p})_{p\in\mathbb{N}_0}$}

In order to calculate our approximation, we will consider the following approach, 
which was also used in \cite{ubpf}. Given any
$(l,P)\in\mathbb{N}_0^2$, we will run Algorithm \ref{alg:est_const} 
in order to obtain $(\Xi_{\theta}^{l,p})_{p\in\{0,1,\dots,P\}}$.
\begin{algorithm}[h!]
\begin{enumerate}
\item Sample: 
Run Algorithm \ref{alg:mlsmc} 
{\em independently} with $N_p-N_{p-1}$ samples
for $p\in\{0,1,\dots,P\}$, up-to level $(l-1)\vee 0$, 
where we define for convenience $N_{-1}:=0$.  
\item Estimate: construct $\Xi_{\theta}^{l,p}$ as in equation \eqref{eq:xi_l_p_def},
for $p\in\{0,1,\dots,P\}$.
\end{enumerate}
\caption{Approach to construct $(\Xi_{\theta}^{l,p})_{p\in\{0,1,\dots,P\}}$ 
for $(l,P)\in\mathbb{N}_0^2$ given.}
\label{alg:est_const}
\end{algorithm}

The joint probability law of the samples simulated 
according to Algorithm \ref{alg:est_const} is 
\begin{equation}\label{eq:alg_2_law}
\mathbb{P}_{\theta}\big(d(u_0^{1:N_p},\dots,u_{(l-1)\vee 0}^{1:N_p})\big) = 
\prod_{p=0}^P \mathsf{P}_\theta^{N_p-N_{p-1}}
\big((u_0^{N_{p-1}+1:N_p},\dots,u_{(l-1)\vee 0}^{N_{p-1}+1:N_p})\big) \, ,
\end{equation}
where $N_{-1}=0$ and $\mathsf{P}_\theta^{N_p-N_{p-1}}$ is as 
defined in \eqref{eq:mlsmc_law}. For 
$(l,P)\in\mathbb{N}_0^2$ given, 
consider running Algorithm \ref{alg:est_const}. 
Then for any $s\in\{0,1,\dots,(l-1)\vee 0\}$ 
and any $p \in \{0,\dots, P\}$ we can construct the following 
empirical 
probability measure on $(\mathsf{X},\mathcal{X})$ 
\begin{equation}\label{eq:es}
\eta_\theta^{s,N_{0:p}}(du_s) := \sum_{q=0}^p \Big(\frac{N_q-N_{q-1}}{N_p}\Big)\eta_\theta^{s,N_q-N_{q-1}}(du_s) \, .
\end{equation}
Note the recursion
$$
\eta_\theta^{s,N_{0:p}}(du_s) = 
\Big(\frac{N_p-N_{p-1}}{N_p}\Big)\eta_\theta^{s,N_p-N_{p-1}}(du_s) 
+ \frac{N_{p-1}}{N_p}\eta_\theta^{s,N_{0:p-1}}(du_s)
\, .
$$
Now define 
\begin{equation}\label{eq:xi_l_p_def}
\Xi_{\theta}^{l,p} := \left\{\begin{array}{ll} 
\eta_{\theta}^{0,N_{0:p}}(\varphi_{\theta}^0) - \eta_{\theta}^{0,N_{0:p-1}}(\varphi_{\theta}^0)
 & \textrm{if}~l=0 \\
\frac{\eta_{\theta}^{l-1,N_{0:p}}(G_{\theta}^{l-1}\varphi_{\theta}^l)}{\eta_{\theta}^{l-1,N_{0:p}}(G_{\theta}^{l-1})} - \eta_{\theta}^{l-1,N_{0:p}}(\varphi_{\theta}^{l-1})
-
\Big(\frac{\eta_{\theta}^{l-1,N_{0:p-1}}(G_{\theta}^{l-1}\varphi_{\theta}^l)}{\eta_{\theta}^{l-1,N_{0:p-1}}(G_{\theta}^{l-1})} - \eta_{\theta}^{l-1,N_{0:p-1}}(\varphi_{\theta}^{l-1})\Big)
& \textrm{otherwise} \, ,
\end{array}\right .
\end{equation}
where $\eta_{\theta}^{0,N_{0:-1}}(\varphi_{\theta}^0):=0$, and 
$$
\frac{\eta_{\theta}^{l-1,N_{0:-1}}(G_{\theta}^{l-1}\varphi_{\theta}^l)}{\eta_{\theta}^{l-1,N_{0:-1}}(G_{\theta}^{l-1})} - \eta_{\theta}^{l-1,N_{0:-1}}(\varphi_{\theta}^{l-1}) := 0 \, .
$$

For convenience in the next section, the conditions \eqref{eq:cs_fv1}-\eqref{eq:cs_fv2} translated to the notations used in this section are
\begin{eqnarray}
\sum_{p\in\mathbb{N}_0}\frac{1}{\overline{\mathbb{P}}_P(p)}\mathbb{E}_{\theta}[\|[\eta_\theta^{0,N_{0:p}}-\eta_\theta^{0}](\varphi_\theta^{0})\|^2] & < &+\infty 
\label{eq:cs_fv3}\\
\sum_{p\in\mathbb{N}_0}\frac{1}{\overline{\mathbb{P}}_P(p)} 
\mathbb{E}_{\theta}\Big[\Big\|\frac{ \eta_\theta^{l-1,N_{0:p}}(G_{\theta}^{l-1}\varphi_\theta^{l}) }
{\eta_\theta^{l-1,N_{0:p}}(G_{\theta}^{l-1})}- \eta_\theta^{l-1,N_{0:p}}(\varphi_\theta^{l-1}) -
\Big(
\frac{ \eta_\theta^{l-1}(G_{\theta}^{l-1}\varphi_\theta^{l}) }
{\eta_\theta^{l-1}(G_{\theta}^{l-1})}- \eta_\theta^{l-1}(\varphi_\theta^{l-1})
\Big)
\Big\|^2\Big]
& < &+\infty , \,\, l\in\mathbb{N} \, , \label{eq:cs_fv4}
\end{eqnarray}
where $\mathbb{E}_{\theta}$ 
is used to denote expectation associated to the probability $\mathbb{P}_{\theta}$ in \eqref{eq:alg_2_law}.

\subsection{Method}

The new method is now 
presented in Algorithm \ref{alg:method}. 
\begin{algorithm}[h!]
\vspace{5pt}
{For $i=1,\dots, M$:}
\begin{enumerate}
\item{Generate $L_i\sim\mathbb{P}_L$ and $P_i\sim\mathbb{P}_P$.}
\item{Run Algorithm \ref{alg:est_const} with $l=L_i$ and $P=P_i$.}
\item{Compute:
$$
\Xi_{\theta}^{L_i} = \sum_{p=0}^{P_i} \frac{1}{\overline{\mathbb{P}}_P(p)}\Xi_{\theta}^{L_i,p}
\, ,
$$
where $\Xi_{\theta}^{L_i,p}$ is given in \eqref{eq:xi_l_p_def}.} 
\end{enumerate}
{Return the estimate:
\begin{equation}\label{eq:ub_est}
\widehat{\eta_{\theta}(\varphi_{\theta})} := \frac{1}{M}\sum_{i=1}^M \frac{1}{\mathbb{P}_L(L_i)}\Xi_{\theta}^{L^i} \, .
\end{equation}
}
\caption{Method for Unbiasedly Estimating $\eta_{\theta}(\varphi_{\theta})$.}
\label{alg:method}
\end{algorithm}

The estimate of
$\eta_{\theta}(\varphi_{\theta})$ is given by \eqref{eq:ub_est}. 
In Section \ref{sec:theory} we will prove that it is both unbiased and of finite variance,
as well as to investigate the cost of computing the estimate.

There are several points of practical interest to be made at this stage 
(the first two were noted already in \cite{ubpf}). 
First, the loop over the number of
independent samples $i$ in Algorithm \ref{alg:method} can be easily parallelized. 
Second, one does not need to make $L$ and $P$ independent;
this is only assumed for simplicity of presentation, but is not required. 
Third, the current method uses only the level $l-1$
marginal of \eqref{eq:alg_2_law}.
All the samples for $s=0,\dots, l-2$ 
and associated empirical measures \eqref{eq:es} are discarded
and only the level $l-1$ empirical measure is utilized. 
This differs from \cite{beskos} where 
all the lower level empirical measures are used. 
It is possible these samples could be utilized to improve the accuracy of the method,
but it is not necessary and so is not investigated further here.
The potential efficiency of the double randomization scheme, as well as a
discussion of the overall efficiency of the approach, 
is given in \cite[Section 2.5]{ubpf}.

\section{Theoretical Results}\label{sec:theory}

Our main objective is to show that $(\Xi_\theta^l)_{l\in\mathbb{N}_0}$ as defined in \eqref{eq:xi_l_def} with $(\Xi_\theta^{l,p})_{p\in\mathbb{N}_0}$
as in \eqref{eq:xi_l_p_def} will satisfy \eqref{eq:ub1}-\eqref{eq:ub3}. To that end, one must first show that $(\Xi_\theta^{l,p})_{p\in\mathbb{N}_0}$ 
satisfy \eqref{eq:cs_fv3}-\eqref{eq:cs_fv4} which certainly verifies \eqref{eq:ub1}-\eqref{eq:ub2} and then one must establish that \eqref{eq:ub3} holds.
We make the following assumptions.

\begin{hypA}
\label{hyp:A}
For each $\theta\in\Theta$, there exist $0<\underline{C}<\overline{C}<+\infty$ such that
\begin{eqnarray*}
\sup_{l\geq 0} \sup_{u\in \mathsf{X}} G_{\theta}^l (u) & \leq & \overline{C}\\
\inf_{l\geq 0} \inf_{u\in \mathsf{X}} G_{\theta}^l (u) & \geq & \underline{C}.
\end{eqnarray*}
\end{hypA}

\begin{hypA}
\label{hyp:B}
For each $\theta\in\Theta$, , there exist a $\rho\in(0,1)$ such that for any $l\geq 1$, $(u,v)\in \mathsf{X}^2$, $A\in\mathcal{X}$
$$
\int_A M_{\theta}^l(u,du') \geq \rho \int_A M_{\theta}^l(v,dv').
$$
\end{hypA}

\begin{hypA}
\label{hyp:C}
For each $\theta\in\Theta$, there exists a $\widetilde{C}<+\infty$ such that for each $i\in\{1,\dots,d_{\theta}\}$
$$
\sup_{l\geq 0} \sup_{u\in \mathsf{X}} |(\varphi_{\theta}^l(u))_i| \leq \widetilde{C}.
$$
\end{hypA}

For $\varphi\in\mathcal{B}_b(\mathsf{X})$ we set $\|\varphi\|_{\infty}=\sup_{u\in\mathsf{X}}|\varphi(u)|$.
To simplify our notations we will set $Z_\theta^l=\int_{\mathsf{X}}\gamma_{\theta}^l(u)du$, $l\in\mathbb{N}_0$, and for $l\in\mathbb{N}$
$$
\|{\varphi}_\theta^{l}-{\varphi}_\theta^{l-1}\|_{\infty}^2 = \max_{i\in\{1,\dots,d_{\theta}\}}\Big\{\|(\varphi_\theta^{l})_i-(\varphi_\theta^{l-1})_i\|_{\infty}^2\Big\}.
$$
We begin with the following result, which is associated to verifying that \eqref{eq:cs_fv3}-\eqref{eq:cs_fv4} can hold.

\begin{prop}\label{prop:main_prop}
Assume (A\ref{hyp:A}-\ref{hyp:C}). Then for any $\theta\in\Theta$ there exists a $C<+\infty$ such that 
for any $p\in\mathbb{N}_0$, $1\leq N_0<N_1<\cdots<N_p<+\infty$:
$$
\mathbb{E}_{\theta}[\|[\eta_\theta^{0,N_{0:p}}-\eta_\theta^{0}](\varphi_\theta^{0})\|^2] \leq \frac{C}{N_p}\Big(1+\frac{p^2}{N_p}\Big).
$$
In addition, for any  $(l,p)\in\mathbb{N}\times\mathbb{N}_0$, $1\leq N_0<N_1<\cdots<N_p<+\infty$:
$$
\mathbb{E}_{\theta}\Bigg[\Bigg\|\frac{ \eta_\theta^{l-1,N_{0:p}}(G_{\theta}^{l-1}\varphi_\theta^{l}) }
{\eta_\theta^{l-1,N_{0:p}}(G_{\theta}^{l-1})}- \eta_\theta^{l-1,N_{0:p}}(\varphi_\theta^{l-1}) -
\Big(
\frac{ \eta_\theta^{l-1}(G_{\theta}^{l-1}\varphi_\theta^{l}) }
{\eta_\theta^{l-1}(G_{\theta}^{l-1})}- \eta_\theta^{l-1}(\varphi_\theta^{l-1})
\Big)
\Bigg\|^2\Bigg] \leq 
$$
$$
\frac{C}{N_p}\Big(1+\frac{p^2}{N_p}\Big)\Big(\|{\varphi}_\theta^{l}-{\varphi}_\theta^{l-1}\|_{\infty}^2+\Big\|G_{\theta}^{l-1}\frac{Z_\theta^{l-1}}{Z_\theta^{l}}-1\Big\|_{\infty}^2\Big).
$$
\end{prop}

\begin{proof}
The first result follows by Lemma \ref{lem:est_pool} in the appendix and the second from Lemma \ref{lem:tech_lem2} also in the appendix.
\end{proof}

\begin{rem}
To show that \eqref{eq:cs_fv3}-\eqref{eq:cs_fv4} can hold, one can set, for instance $N_p=2^p$ and $\mathbb{P}_P(p)\propto 2^{-p}(p+1)\log_2(p+2)^2$. 
See for example \cite{ubpf} and \cite{rhee}.
\end{rem}

To continue our discussion, to complete our proof, we must know something about the quantities
$$
\|{\varphi}_\theta^{l}-{\varphi}_\theta^{l-1}\|_{\infty}^2\quad {\rm and}\quad
\Big\|G_{\theta}^{l-1}\frac{Z_\theta^{l-1}}{Z_\theta^{l}}-1\Big\|_{\infty}^2
$$
in terms of a possible decay as a function of $l$. 
To that end, we shall assume that these terms are 
$\mathcal{O}(h_l^{\beta})$ for some $\beta>0$.
This assumption can be verified for the example in Section \ref{sec:example}.
Recall from Section \ref{sec:example} that $h_l=2^{-l}$. 

\begin{prop}\label{prop:verify}
Assume (A\ref{hyp:N}). 
Then there is $C > 0$,  
depending on $f$ and $u_*$, and $\beta>0$ 
such that for all $u\in \mathsf{X}$
$$
\| p^l(\cdot;u) \|, \| p(\cdot;u) \| <C \, , \qquad 
\| p^l(\cdot;u) - p(\cdot;u) \|^2 \leq C h_l^\beta \, ,
$$
where the norm is $L^2(D)$.
Given a function $F : \mathbb{N} \times \mathsf{X} \rightarrow \bbR^n$, 
suppose that there is a $C'>0$ which does not depend on $(l,u)$
such that
\begin{equation}\label{eq:cty}
\|F(l,u) - F(\infty,u)\| \leq C' \| p^l(\cdot ; u) - p(\cdot ; u) \| \, ,
\end{equation} 
where the first norm is understood as the $n-$dimensional 
Euclidean norm, while the second norm is $L^2(D)$,
and $F(\infty, \cdot) := \lim_{l\rightarrow \infty} F(l,\cdot)$.
Then there is another $C>0$ which does not depend on $(l,u)$ such that
$$\|F(l,\cdot ) - F(l-1,\cdot )\|_\infty^2 \leq C h_l^\beta.$$
\end{prop}

\begin{proof}
This is a slight generalization of the results of \cite{beskos}, Sec. 4,
where it was verified that 
$\Big\|G_{\theta}^{l-1}\frac{Z_\theta^{l-1}}{Z_\theta^{l}}-1\Big\|_{\infty}^2 = 
\mathcal{O}(h_l^\beta)$.
It is well known that for $v\in \mathbb{R}^n$, there is a $C>0$ such that 
$\|v\|_\infty \leq C \|v\|$. The result follows by taking supremum over $u$.
\end{proof}

Note that 
$G_{\theta}^{l-1}\frac{Z_\theta^{l-1}}{Z_\theta^{l}} = 
\frac{G_{\theta}^{l-1}}{\eta_\theta^{l-1}(G_{\theta}^{l-1})}$ and $G_\theta^{\infty}=1$.
So
$$
|G_{\theta}^{l-1}\frac{Z_\theta^{l-1}}{Z_\theta^{l}}-1| 
\leq \frac2{\eta_\theta^{l-1}(G_{\theta}^{l-1})}|G_{\theta}^{l-1}-1| \, .
$$
Defining $F(l,u) = G_\theta^{l}(u)$ then  
assumption (A\ref{hyp:A}) and Prop. 4.1 of \cite{beskos}
together with Proposition \ref{prop:verify} imply there is a $C>0$ such that
$$
\Big\|G_{\theta}^{l-1}\frac{Z_\theta^{l-1}}{Z_\theta^{l}}-1\Big\|_{\infty}^2 \leq C h_l^\beta \, .
$$
Defining $F(l,u) := {\varphi}_\theta^{l}(u)$, as defined in \eqref{eq:phil} and \eqref{eq:phi},
then it is easy to show that Proposition \ref{prop:verify} ensures 
$$
\|{\varphi}_\theta^{l}-{\varphi}_\theta^{l-1}\|_{\infty}^2 \leq C h_l^\beta \, .
$$
See equation (19) of \cite{beskos}. Assumptions (A\ref{hyp:A}) and (A\ref{hyp:C})
are similarly verified.

\begin{theorem}
Assume (A\ref{hyp:A}-\ref{hyp:C}).  Then there exist choices of $\mathbb{P}_L,\mathbb{P}_P$ and $(N_p)_{p\in\mathbb{N}_0}$, 
$1\leq N_0<N_1<\cdots$ so that $(\Xi_\theta^l)_{l\in\mathbb{N}_0}$ as defined in \eqref{eq:xi_l_def} with $(\Xi_\theta^{l,p})_{p\in\mathbb{N}_0}$
as in \eqref{eq:xi_l_p_def} will satisfy \eqref{eq:ub1}-\eqref{eq:ub3}. That is, \eqref{eq:ub_est} is an unbiased and finite variance estimator
of $\eta_{\theta}(\varphi_{\theta})$.
\end{theorem}

\begin{proof}
Throughout the proof $C$ is a finite constant that will not depend on $l$ or $p$ and whose value will change upon each appearance.
Given the commentary above, we need only show that \eqref{eq:ub3} can hold for some given choices of of $\mathbb{P}_L,\mathbb{P}_P$ and $(N_p)_{p\in\mathbb{N}_0}$.
We have that
\begin{equation}\label{eq:main_theo_eq1}
\sum_{l\in\mathbb{N}_0}\frac{1}{\mathbb{P}_L(l)}\mathbb{E}_{\theta}[\|\Xi_{\theta}^l\|^2]  = 
\sum_{(l,p)\in\mathbb{N}_0^2}\frac{\mathbb{P}_P(p)}{\mathbb{P}_L(l)}\Big\{
\sum_{s=0}^p \frac{\mathbb{E}_{\theta}[\|\Xi_{\theta}^{l,s}\|^2]}{\overline{\mathbb{P}}_P(s)^2} + 2\sum_{0\leq s <q\leq p}
\frac{\mathbb{E}_{\theta}[\|\Xi_{\theta}^{l,s}\|\|\Xi_{\theta}^{l,q}\|]}{\overline{\mathbb{P}}_P(p)\overline{\mathbb{P}}_P(q)}
\Big\}.
\end{equation}
Now recalling \eqref{eq:xi_l_p_def} and noting that for $p\in\mathbb{N}$
$$
\eta_{\theta}^{0,N_{0:p}}(\varphi_{\theta}^0) - \eta_{\theta}^{0,N_{0:p-1}}(\varphi_{\theta}^0) =
\eta_{\theta}^{0,N_{0:p}}(\varphi_{\theta}^0) - \eta_{\theta}^{0}(\varphi_{\theta}^0)  - \{\eta_{\theta}^{0,N_{0:p-1}}(\varphi_{\theta}^0)-\eta_{\theta}^{0}(\varphi_{\theta}^0)\}
$$
and that for $p\in\mathbb{N}$
$$
\frac{\eta_{\theta}^{l-1,N_{0:p}}(G_{\theta}^{l-1}\varphi_{\theta}^l)}{\eta_{\theta}^{l-1,N_{0:p}}(G_{\theta}^{l-1})} - \eta_{\theta}^{l-1,N_{0:p}}(\varphi_{\theta}^{l-1})
-
\Big(\frac{\eta_{\theta}^{l-1,N_{0:p-1}}(G_{\theta}^{l-1}\varphi_{\theta}^l)}{\eta_{\theta}^{l-1,N_{0:p-1}}(G_{\theta}^{l-1})} - \eta_{\theta}^{l-1,N_{0:p-1}}(\varphi_{\theta}^{l-1})\Big) = 
$$
$$
\frac{\eta_{\theta}^{l-1,N_{0:p}}(G_{\theta}^{l-1}\varphi_{\theta}^l)}{\eta_{\theta}^{l-1,N_{0:p}}(G_{\theta}^{l-1})} - \eta_{\theta}^{l-1,N_{0:p}}(\varphi_{\theta}^{l-1}) -
\Big(
\frac{ \eta_\theta^{l-1}(G_{\theta}^{l-1}\varphi_\theta^{l}) }
{\eta_\theta^{l-1}(G_{\theta}^{l-1})}- \eta_\theta^{l-1}(\varphi_\theta^{l-1})
\Big) -
$$
$$
\Big\{
\Big(\frac{\eta_{\theta}^{l-1,N_{0:p-1}}(G_{\theta}^{l-1}\varphi_{\theta}^l)}{\eta_{\theta}^{l-1,N_{0:p-1}}(G_{\theta}^{l-1})} - \eta_{\theta}^{l-1,N_{0:p-1}}(\varphi_{\theta}^{l-1})\Big) -
\Big(
\frac{ \eta_\theta^{l-1}(G_{\theta}^{l-1}\varphi_\theta^{l}) }
{\eta_\theta^{l-1}(G_{\theta}^{l-1})}- \eta_\theta^{l-1}(\varphi_\theta^{l-1})
\Big)
\Big\}
$$
by Proposition \ref{prop:main_prop} that for $(l,s)\in\mathbb{N}_0^2$
\begin{equation}\label{eq:main_theo_eq2}
\mathbb{E}_{\theta}[\|\Xi_{\theta}^{l,s}\|^2] \leq \frac{Ch_l^{\beta}}{N_s}\Big(1+\frac{s^2}{N_s}\Big).
\end{equation}
Also, by using Cauchy-Schwarz, $(l,p,q)\in\mathbb{N}_0\times\mathbb{N}^2$
\begin{equation}\label{eq:main_theo_eq3}
\mathbb{E}_{\theta}[\|\Xi_{\theta}^{l,s}\|\|\Xi_{\theta}^{l,q}\|] \leq \frac{Ch_l^{\beta}}{N_s^{1/2}N_q^{1/2}}\Big(1+\frac{s^2}{N_s}\Big)^{1/2}\Big(1+\frac{q^2}{N_q}\Big)^{1/2}.
\end{equation}
Then using the bounds \eqref{eq:main_theo_eq2}-\eqref{eq:main_theo_eq3} in \eqref{eq:main_theo_eq1} gives the upper-bound (noting that the
case $s=0=q$ the terms $\mathbb{E}_{\theta}[\|\Xi_{\theta}^{l,s}\|^2]$ and $\mathbb{E}_{\theta}[\|\Xi_{\theta}^{l,s}\|\|\Xi_{\theta}^{l,q}\|]$ are $\mathcal{O}(1)$ so one can find a $C$ such that the following upper-bound holds)
\begin{eqnarray*}
\sum_{l\in\mathbb{N}_0}\frac{1}{\mathbb{P}_L(l)}\mathbb{E}_{\theta}[\|\Xi_{\theta}^l\|^2]  & \leq &
C\sum_{(l,p)\in\mathbb{N}_0^2}\frac{\mathbb{P}_P(p)h_l^{\beta}}{\mathbb{P}_L(l)}\Bigg\{
\sum_{s=0}^p  \frac{\Big(1+\frac{s^2}{N_s}\Big)}{N_s\overline{\mathbb{P}}_P(s)^2} + 
\sum_{0\leq s <q\leq p}
\frac{\Big(1+\frac{s^2}{N_s}\Big)^{1/2}\Big(1+\frac{q^2}{N_q}\Big)^{1/2}}{N_s^{1/2}N_q^{1/2}\overline{\mathbb{P}}_P(p)\overline{\mathbb{P}}_P(q)}
\Bigg\}.
\end{eqnarray*}
Now if one chooses, for instance $N_p=2^p$ and $\mathbb{P}_P(p)\propto 2^{-p}(p+1)\log_2(p+2)^2$ and $\mathbb{P}_L(l)\propto h_l^{\alpha\beta}$ for any $\alpha\in(0,1)$
then \eqref{eq:ub3} is satisfied and hence the proof is completed.
\end{proof}

In most cases of practical interest,  it is not possible to choose 
$\mathbb{P}_L,\mathbb{P}_P$ and $(N_p)_{p\in\mathbb{N}_0}$ so that 
\eqref{eq:ub_est} is an unbiased and finite variance estimator, as well as having finite expected cost. Suppose, in the case of 
Section \ref{sec:example}, 
the cost to evaluate $G_{\theta}^l$ is $\mathcal{O}(h_l^{-1})$ and $\beta=1$. 
Then, just as in \cite{ubpf}, if we 
choose $\mathbb{P}_L(l)\propto 2^{-l}(l+1)\log_2(l+2)^2$, $N_p=2^p$, 
and $\mathbb{P}_P(p)\propto 2^{-p}(p+1)\log_2(p+2)^2$,
then to achieve a variance of $\mathcal{O}(\epsilon^2)$ (for $\epsilon>0$ arbitrary) the cost is $\mathcal{O}(\epsilon^{-2}|\log(\epsilon)|^{2+\delta})$
for any $\delta>0$, with high probability.
For the MLSMC method in \cite{beskos}, the cost to obtain a mean square error of 
$\mathcal{O}(\epsilon^2)$ is
$\mathcal{O}(\epsilon^{-2}\log(\epsilon)^2)$, 
which is a mild reduction in cost. 
However, we note that this discussion is constrained to the case of a single processor.
The unbiased method is straightforwardly parallelized.

\section{Numerical Results}\label{sec:numerics}

First we will consider a toy example where we can analytically calculate the marginal likelihood
and investigate the performance of the resulting estimator in comparison to the 
estimator obtained using the original MLSMC method of \cite{beskos} (not presented here). 
Subsequently we will consider the example from Section \ref{sec:example}.
Finally, for both examples 
we will explore the potential applicability of our estimators within the
context of parameter optimization using the 
stochastic gradient descent (SGD) method.

The forward model is the same for both problems, 
hence the anticipated rate of convergence is the same, and is estimated 
as $\beta=4$, just as in \cite{beskos}. 
The cost would be $\mathcal{O}(h_l^{-\gamma})$ in general for the problem in Section \ref{sec:example}, 
and for the particular setup described in \ref{ssec:example} $\gamma=1$.
We choose $h_l = 2^{-l}$ and $\bbP_L(l) \propto 2^{-2.5 l}$.
This is far into the so-called canonical regime ($\beta>\gamma$), 
and therefore we allow unbounded $L$, i.e. $L_{\rm max}=\infty$ in the 
terminology of \cite{ubpf}. 
The reason for this is basically that the sum \eqref{eq:ub3} and the corresponding 
cost series both easily converge, if the cost is deterministic and 
$\mathcal{O}(h_l^{-\gamma})$ as a function of $h_l$.
However, in this case the cost depends 
upon the randomized estimator of the series in $p$.
Since the rate of convergence is borderline in the $p$ direction, 
$\beta_p=1=\gamma_p$, as in \cite{ubpf} we impose a maximum $P_{\rm max}$ on $P$.
This is necessary to prevent the possibility of the algorithm getting stuck with 
an extremely expensive sample. It is discussed further in that work. 
In particular, we choose $N_p=2^{p+3}$ and 
$$
\mathbb{P}_P(p) \propto 
\mathbb{I}(0\leq p\leq P_{\rm max}) 
\left\{\begin{array}{ll}
		2^{-p+4}
		& \textrm{if}~p<4, \\
		2^{-p}\cdot p \cdot \log_2(p)^2 & \textrm{otherwise} \, .
		\end{array}\right  .
		$$
The piecewise definition of $\mathbb{P}_P$ ensures that it has the correct asymptotic behaviour but is also monotonically decreasing. Note that in this regime, i.e. strongly canonical convergence in $L$, or large $\beta >\gamma$, the MLSMC method 
easily achieves the optimal complexity $\mathcal O(\epsilon^{-2})$. 
However, since the convergence rate in $P$ is necessarily subcanonical,  
our method therefore suffers from a logarithmic penalty, i.e.
$\mathcal O(\epsilon^{-2}\log(\epsilon)^{2+\delta})$, for any $\delta>0$. 
This cannot be observed in the simulations though. 
Empirically we observe that we can set $P_{\rm max}$ rather small, 
which is perhaps afforded by the very fast convergence in the $L$ direction.
This may also be why we cannot see the theoretically predicted log penalty 
in the simulations.

\subsection{Toy example}\label{ssec:toy}

We first consider an example where the marginal likelihood is analytically calculable. Consider the following PDE on $D$:
\begin{align*}
\nabla^2p  &=u,\quad \textrm{ on } D,       \\
p&= 0, \quad \textrm{ on } \partial D,
\end{align*}
where $D=[0,1]$. The solution of this PDE is $p(x;u)=\frac{u}{2}(x^2-x)$.
Define the observation operator as 
$$\mathcal{G}(u)=[p(x_1;u),p(x_2;u), \dots, p(x_M;u)]^{\intercal}\triangleq Gu \, .$$ 
Suppose the observation takes the form 
$y = \mathcal{G}(u) + \xi, \xi \sim N(0,\theta^{-1}\cdot\bm{I}_M), \xi \perp u$, 
and $\theta$ follows a log-normal distribution with mean 0 and variance $\sigma^2$. 
Then the unnormalized likelihood is given by 
$$
\gamma(u,\theta) = \theta^{\frac{M}{2}}\cdot\exp\left\{-\frac{\theta}{2}\|Gu - y\|^2 \right\} \cdot \frac{1}{\theta}\exp\Big\{ -\frac{(\log(\theta))^2} {2\sigma^2} \Big\} , 
$$
and the marginal likelihood is 
\begin{align*}
\gamma(\theta)  &  = \int_{-1}^1 \gamma(u,\theta) \textrm{d}u \\
& =  \theta^{\frac{M-2}{2}} \exp\Big\{ -\frac{(\log(\theta))^2} {2\sigma^2} \Big\} \int_{-1}^{1} \exp\left\{-\frac{\theta }{2} \|Gu - y\|^2 \right\} \textrm{d}u \\ 
& = \theta^{\frac{M-2}{2}} \exp\left\{ -\frac{\theta}{2} \Big(\|y\|^2 - \frac{(G^{\intercal}y)^2}{\|G\|^2}\Big)-\frac{(\log(\theta))^2} {2\sigma^2} \right\} \int_{-1}^{1} \exp \left\{ -\frac{\theta\|G\|^2}{2} \Big(u - \frac{G^{\intercal}y}{\|G\|^2 } \Big)^2  \right\} \textrm{d}u \\
& = \theta^{\frac{M-3}{2}} \frac{\sqrt{\pi/2}}{\|G\|} \exp\left\{ -\frac{\theta}{2} \Big(\|y\|^2 - \frac{(G^{\intercal}y)^2}{\|G\|^2}\Big) -\frac{(\log(\theta))^2} {2\sigma^2} \right\}  \Bigg( \textrm{erf}\Big(\sqrt{\frac{\theta}{2}}\|G\|(1- \frac{G^{\intercal}y}{\|G\|^2 })\Big) - \\
&\qquad\qquad\qquad\qquad\qquad\qquad\qquad \qquad\qquad\qquad\qquad\qquad~~~  \textrm{erf}\Big(\sqrt{\frac{\theta}{2}}\|G\|(-1- \frac{G^{\intercal}y}{\|G\|^2 })\Big) \Bigg). 
\end{align*}
So the log-likelihood is then given by
\begin{align*}
\log\big(\gamma(\theta)\big)   = &~ \frac{M-3}{2}\log(\theta) - \frac{\theta}{2} \Big(\|y\|^2 -  \frac{(G^{\intercal}y)^2}{\|G\|^2}\Big) -\frac{(\log(\theta))^2} {2\sigma^2}  + \log \Bigg( \textrm{erf}\Big(\sqrt{\frac{\theta}{2}}\|G\|(1- \frac{G^{\intercal}y}{\|G\|^2 })\Big) - \\ 
&~  \textrm{erf}\Big(\sqrt{\frac{\theta}{2}}\|G\|(-1- \frac{G^{\intercal}y}{\|G\|^2 })\Big) \Bigg) + C,  
\end{align*}
and the derivative of the log-likelihood is 
\begin{align*}
\frac{\partial\log\big(\gamma(\theta)\big)}{ \partial\theta}  & = \frac{M-3}{2\theta} - \frac{\Big(\|y\|^2 - \frac{(G^{\intercal}y)^2}{\|G\|^2}\Big)}{2} - \frac{\log(\theta)} {\sigma^2 \theta}  + \frac{1}{ \textrm{erf}\Big (\sqrt{\frac{\theta}{2}}\|G\|(1- \frac{G^{\intercal}y}{\|G\|^2 })\Big) - \textrm{erf}\Big(\sqrt{\frac{\theta}{2}}\|G\|(-1- \frac{G^{\intercal}y}{\|G\|^2 })\Big)} \cdot \\
& \quad~ \frac{2}{\sqrt{\pi}}\Bigg( \exp\Big\{ -\frac{\theta\|G\|^2}{2}(1- \frac{G^{\intercal}y}{\|G\|^2 })^2 \Big\} \frac{\|G\|(1- \frac{G^{\intercal}y}{\|G\|^2 })}{2\sqrt{2\theta}} - \\
& \qquad\qquad \exp\Big\{ -\frac{\theta\|G\|^2}{2}(-1- \frac{G^{\intercal}y}{\|G\|^2 })^2 \Big\} \frac{\|G\|(-1- \frac{G^{\intercal}y}{\|G\|^2 })}{2\sqrt{2\theta}}  \Bigg).
\end{align*}

First, the performance of the unbiased algorithm for a single gradient estimation 
$\frac{\partial\log\big(\gamma(\theta)\big)}{ \partial\theta}\Big|_{\theta=\theta^*}$ is verified. The data is generated with M=50, and observation operator $\mathcal{G}(u)=[p(x_1;u),p(x_2;u),$ $\dots, p(x_M;u)]^{\intercal}$ with $x_i = i/(M+1)$.  $\theta^*$ is set to be $2$, and the true value of the derivative of the log-likelihood at this point is calculated using the above analytical solution. 
For each $L$, the MLSMC estimator is realized 50 times and the MSE is reported. 
Similarly, the MSE of unbiased algorithm is calculated based on 50 realizations. 
The results are presented in Figure \ref{fig:single_anal}.
The cost reported in the plot is proportional to 
the sum of the cost per forward solve at level $l$ 
(tridiagonal linear system), $h_l^{-1}$, 
multiplied by the total number of samples at level $l$.

\begin{figure}[!htbp]
	\centering\includegraphics[width=\textwidth]{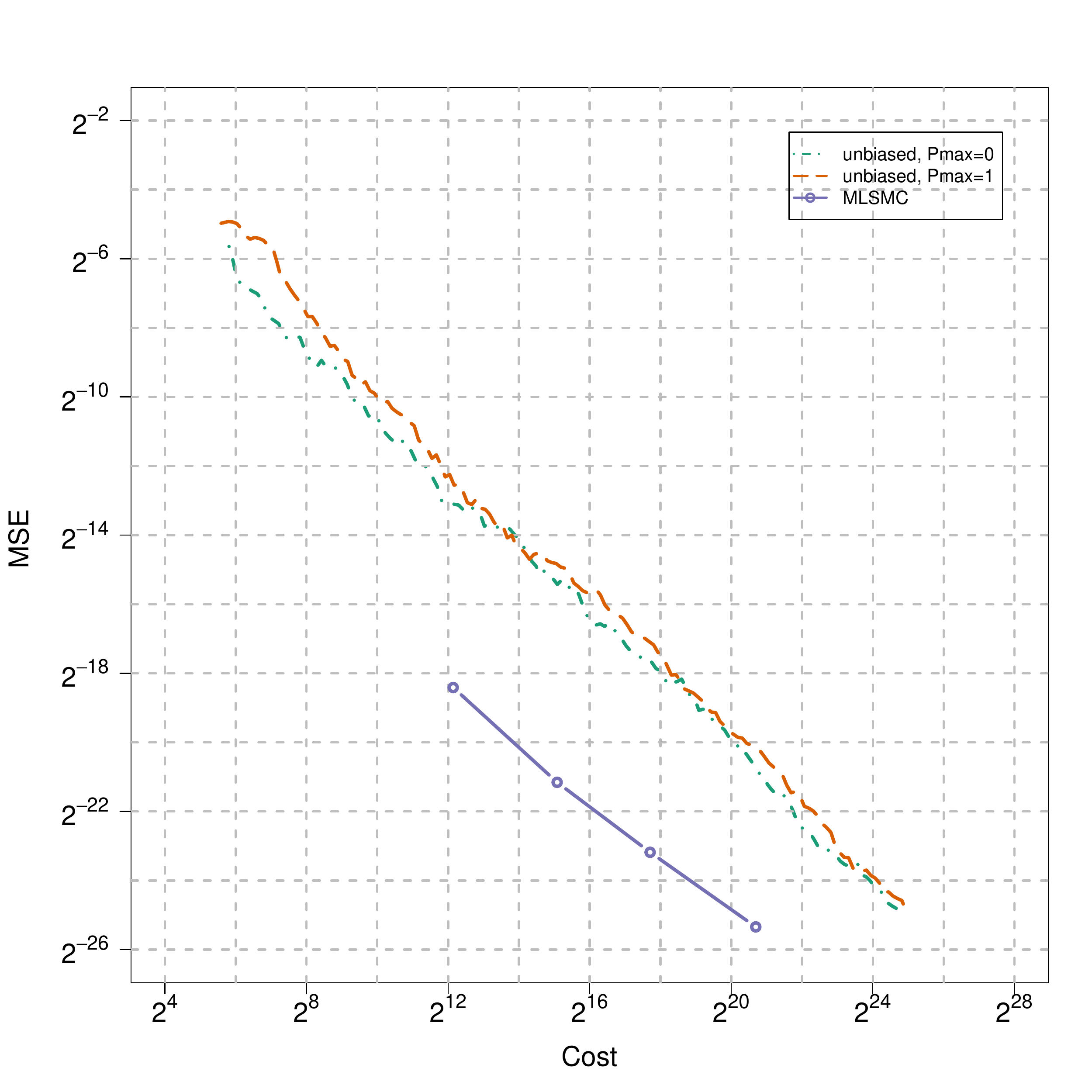}
	\caption{Toy example, single estimator. 
	{MSE vs cost for (i) the unbiased algorithm with different choices of 
	$P_{\rm max}$ and (ii) MLSMC.}}
	\label{fig:single_anal}
\end{figure}


\subsection{Example of Sec. \ref{sec:example}}\label{ssec:full}

Now that we understand the behaviour of the estimator we 
consider the example from Sec. \ref{sec:example}.
Here we do not have an analytical solution, so the true value of the target was first estimated with the MLSMC algorithm with L=12. 
This sampler was realized 50 times and the average of the estimator is taken as the ground truth. 
Now for each $L$, the MLSMC estimator is realized 50 times and the MSE is reported. 
Similarly, the MSE of unbiased algorithm is calculated based on 50 realizations. 
The results are presented in Figure \ref{fig:single_nonanal}.
The cost in the plot is proportional to 
the sum of the cost per forward solve at level $l$ 
(tridiagonal linear system), $h_l^{-1}$, 
multiplied by the total number of samples at level $l$.

\begin{figure}[!htbp]
	\centering\includegraphics[width=\textwidth]{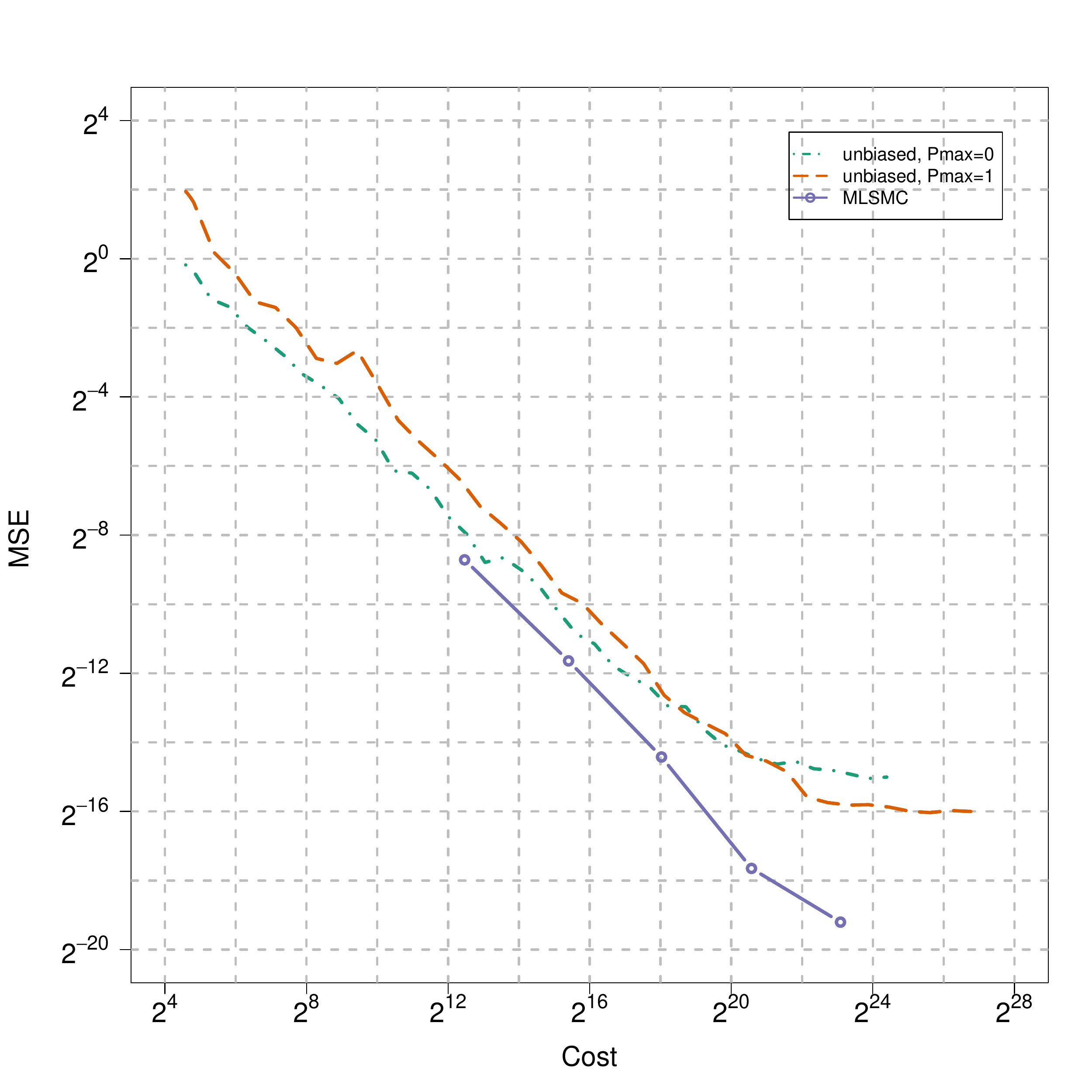}
	\caption{Example of Sec. \ref{sec:example}, single estimator. 
	{MSE vs cost for (i) the unbiased algorithm with different choices of 
	$P_{\rm max}$ and (ii) MLSMC.}}
	\label{fig:single_nonanal}
\end{figure}


\subsection{Stochastic gradient descent method}

In this section, we investigate the potential to use our unbiased estimators within 
the SGD method. 
Recall we want to find the MLE of $\theta$ by minimizing 
$-\log p(y|\theta) = -\log(\int \gamma_{\theta}(u)du)$. 
Our estimator given in equation \eqref{eq:ub_est} 
provides an unbiased estimator $\widehat{\eta_{\theta}(\varphi_{\theta})}$
of $\nabla_\theta \log p(y|\theta)$, for any choice of $M\geq 0$.
In other words $\bbE \widehat{\eta_{\theta}(\varphi_{\theta})} = \nabla_\theta \log p(y|\theta)$. We will see that it is most efficient to choose $M=1$.
To ensure the output of the SGD algorithm satisfies $\theta>0$, 
we let $\theta=\exp(\xi)$ and optimize $\xi$.  
The details are given in Algorithm \ref{algo:sgd}.

\begin{algorithm}[h!]
	\begin{enumerate}
		\item{Initialize $\xi_1$ and choose a sequence $\{\alpha_k\}_{k=1}^\infty$
		and a value $M\in \mathbb{N}$.}
		\item {For $k=1,\dots, K$ (or until convergence criterion is met)} 
			\begin{itemize}
				\item{Compute $\widehat{\eta_{\theta}(\varphi_{e^{\xi_k}})}$
				using \eqref{eq:ub_est}} 
				\item{Update $\xi_{k+1} = \xi_k - 
				\alpha_k \widehat{\eta_{\theta}(\varphi_{e^{\xi_k}})} \exp(\xi_k)$.}
				\end{itemize}
		\item{Return $\theta_{K+1}=\exp(\xi_{K+1})$.}
	\end{enumerate}
	\caption{SGD using new unbiased estimator.}
	\label{algo:sgd}
\end{algorithm}

As above, it makes sense to first explore the toy model with analytical solution.
As before, the MSE is calculated based on 50 realizations, and
the cost in the plot is again proportional to 
the sum of the cost per forward solve at level $l$ 
(tridiagonal linear system), $h_l^{-1}$, 
multiplied by the total number of samples at level $l$.

In Figure \ref{fig:alphakN} we explore the performance of the unbiased estimator
with different choices of $\alpha_k = \alpha_1/k$, $\alpha_1\in\{0.1,0.025\}$,
and different choices of the number of samples $M$ used to construct 
$\widehat{\eta_{\theta}(\varphi_{e^{\xi_k}})}$ using \eqref{eq:ub_est} in 
step 2 of Algorithm \ref{algo:sgd}. The two takeaways from this experiment are 
that (1) it is more efficient to take fewer samples $M$ (in particular $M=1$),
and (2) it is more efficient to choose a larger constant $\alpha_1 = 0.1$.
In particular, the dynamics of the algorithm experiences a phase 
transition as one varies the constant $\alpha_1$. 
A large enough value appears to provide gradient descent type exponential 
convergence $\mathcal O(e^{-\rm cost})$, 
while a value which is too small yields Monte Carlo (MC) type $\mathcal O(1/{\rm cost})$.
 It is notable that the exponential convergence eventually gives way to MC type convergence, 
 and that the point where this occurs increases proportionally to the additional 
 constant in cost incurred with larger sample size $M$, 
 so that the error curves for different values of $M$ eventually intersect.
 
 \begin{figure}[!htbp]
	\centering\includegraphics[width=\textwidth]{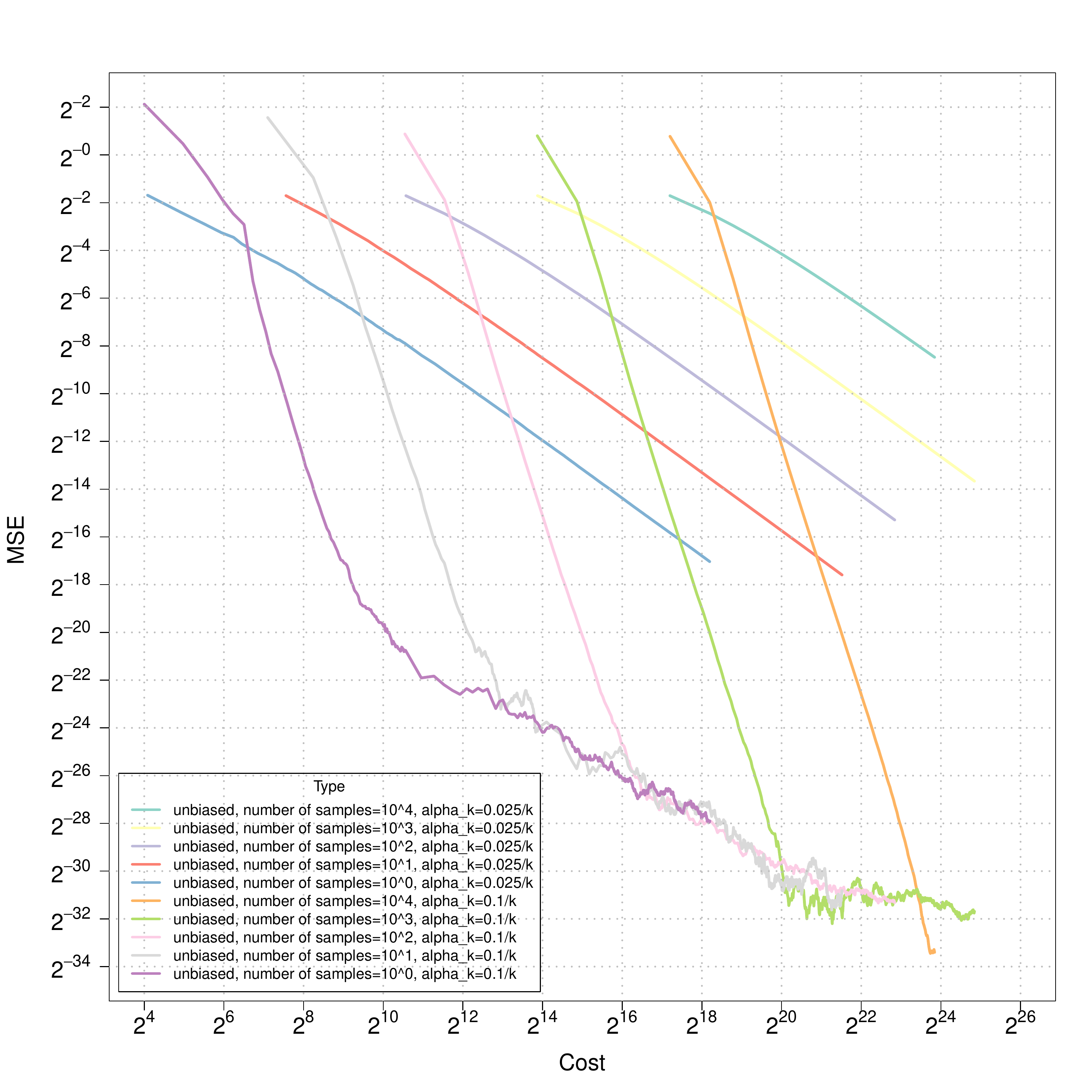}
	\caption{Toy example, SGD. 
	{MSE vs cost for $\alpha_k=0.025/k$ and $\alpha_k=0.1/k$ for a range of 
	sample sizes $M$. $P_{\rm max}=0$ is fixed.}}
	\label{fig:alphakN}
\end{figure}

Natural questions are then whether there is a limit 
to how large one can choose $\alpha_1$ and at which value precisely does the phase transition occur. 
These questions are partially answered by the experiments presented in 
Figure \ref{fig:alphakalpha} (a), where we see that $\alpha_1$ 
should not be chosen larger than $0.2$ and the phase transition happens in between 
$0.025$ and $0.05$. 
Figure \ref{fig:alphakalpha} (b) 
illustrates the benefits and drawbacks of using a constant $\alpha$. 
In particular, the algorithm may converge more quickly at first, 
but plateaus when it reaches the induced bias.

 \begin{figure}[!htbp]
	\centering\includegraphics[width=\textwidth]{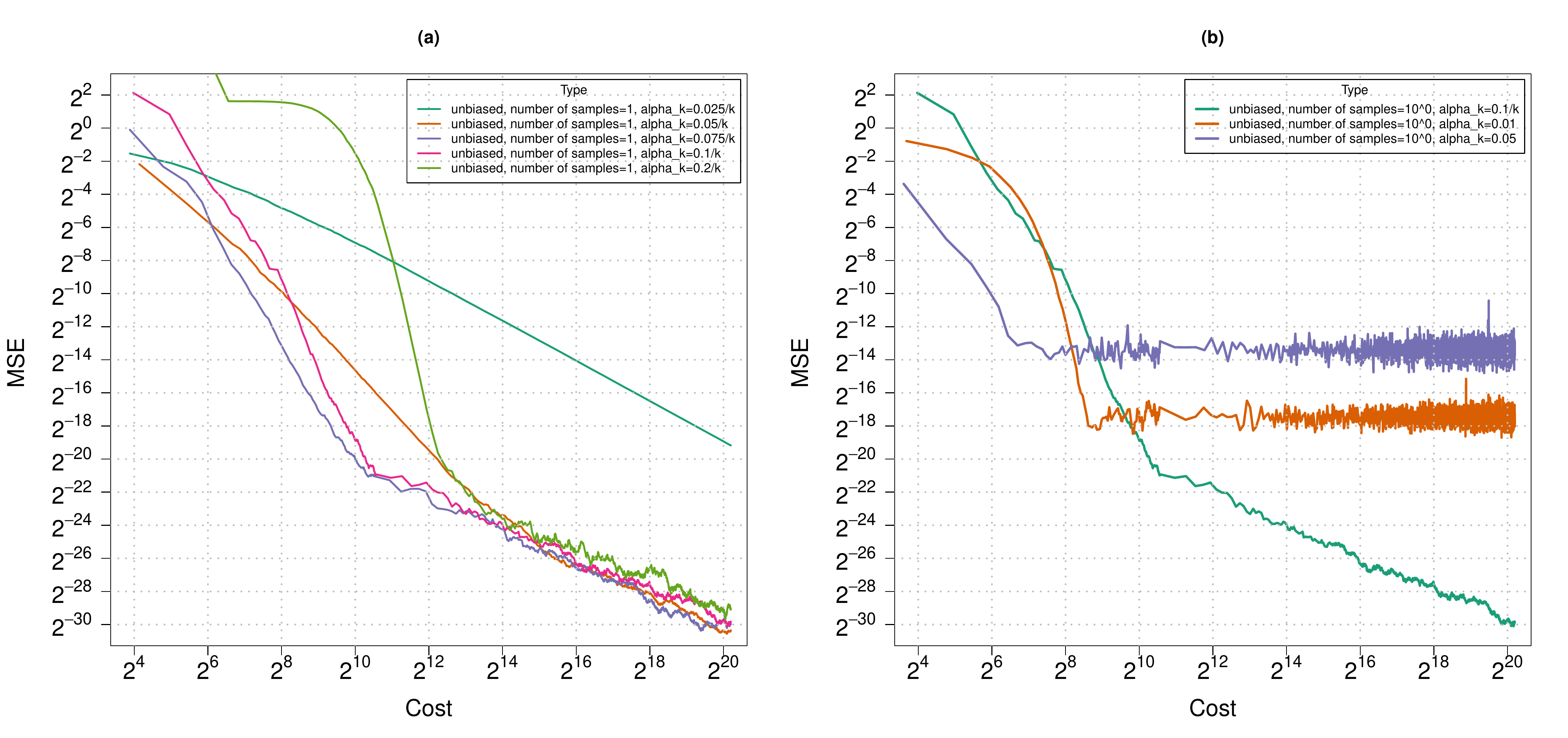}
	\caption{Toy example, SGD. 
	{MSE vs cost for (a) $\alpha_k=\alpha_1/k$ and a range of $\alpha_1$, 
	and (b) some examples of constant $\alpha$. $P_{\rm max}=0$ is fixed.}}
	\label{fig:alphakalpha}
\end{figure}

In Figure \ref{fig:pmax} we explore various choices of $P_{\rm max}$, 
for $M=1$ fixed and $\alpha_1=0.1$. 
It is apparent that it is preferable to choose a smaller value of $P_{\rm max}$.
We note however that there will be an induced bias, which will be larger
for smaller $P_{\rm max}$. 
However, for this particular problem we do not even observe that bias over the 
range of MSE and cost considered.

 \begin{figure}[!htbp]
	\centering\includegraphics[width=\textwidth]{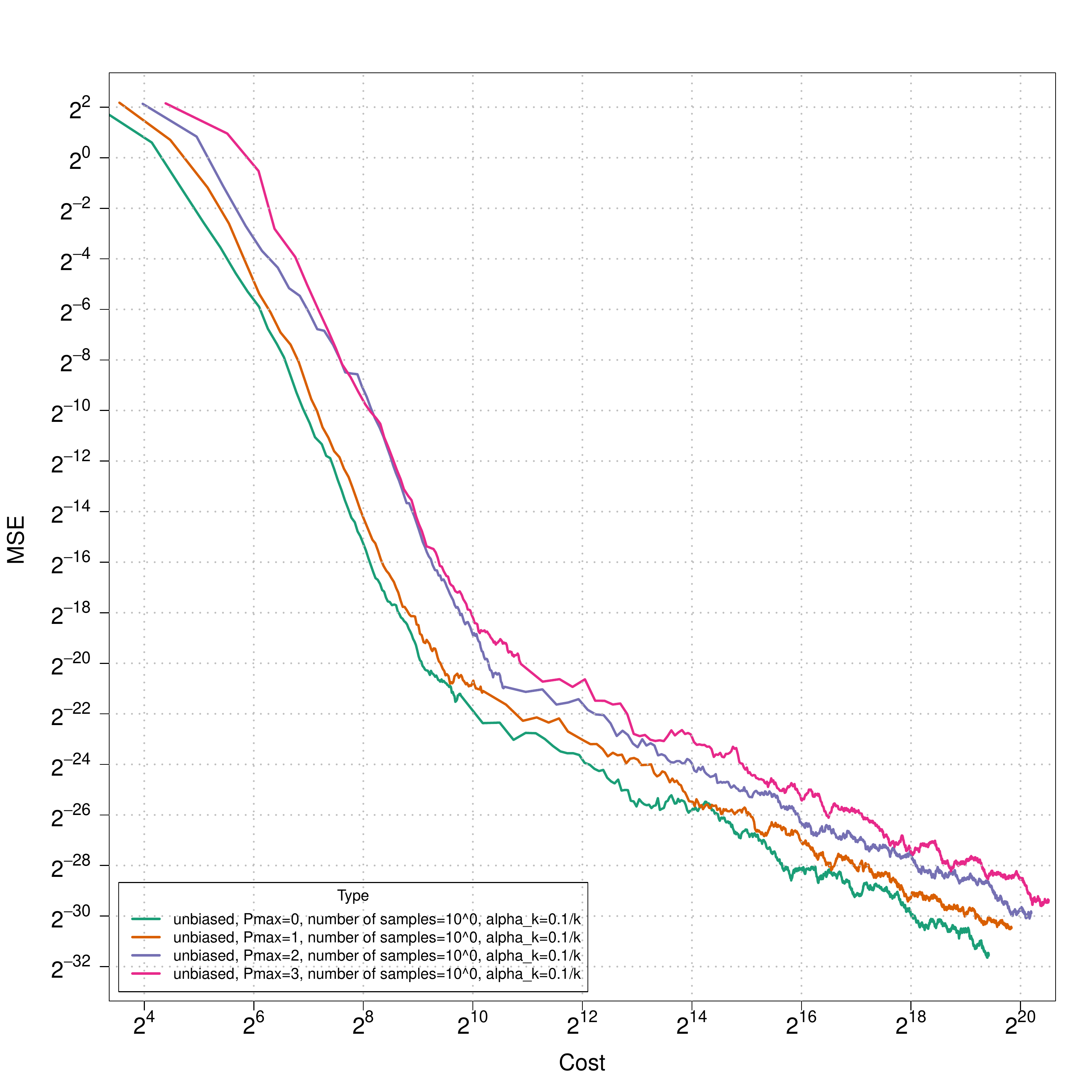}
	\caption{Toy example, SGD.
	{MSE vs cost for $M=1$ fixed and $\alpha_k=0.1/k$, and 
	various choices of $P_{\rm max}$.}}
	\label{fig:pmax}
\end{figure}

As a last experiment with the toy example, we compare the convergence of SGD using 
our unbiased algorithm with $P_{\rm max}=0$, $M=1$, and $\alpha_1=0.1$
to the analogous algorithm where an MLSMC estimator with various $L$ 
(single gradient estimator MSE $\propto 2^{-L}$)
replaces the unbiased estimator in step 2 of Algorithm \ref{algo:sgd}. 
Similar behaviour was observed for MLSMC relative to different choices of $\alpha_k$
as compared to the unbiased estimator. 
The results are shown in Figure \ref{fig:anal}.
Here it is clear that over a wide range of MSE the unbiased estimator 
provides a significantly more efficient alternative to the MLSMC estimator.

 \begin{figure}[!htbp]
	\centering\includegraphics[width=\textwidth]{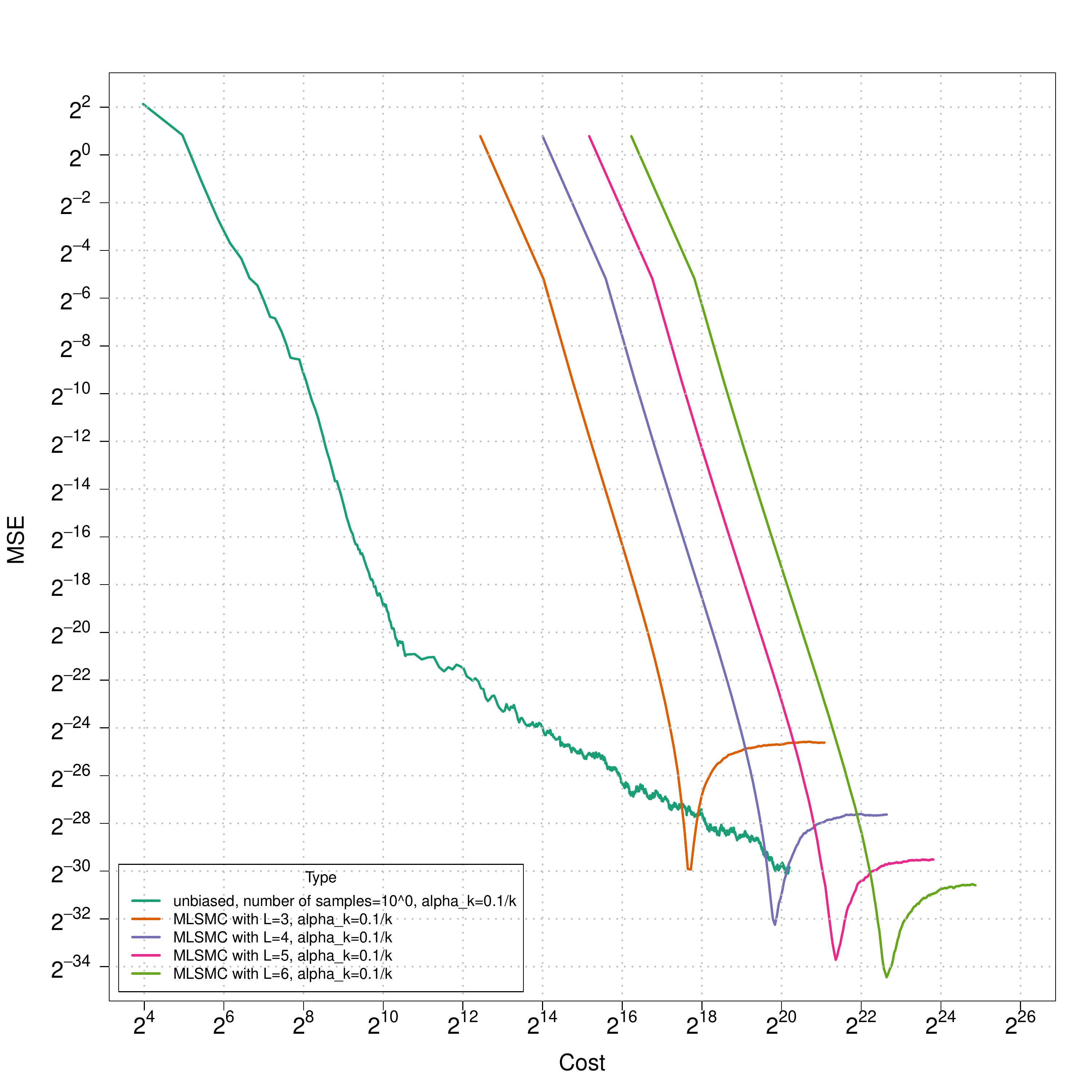}
	\caption{Toy example, SGD. 
	{MSE vs cost for $\alpha_k=0.1/k$, and unbiased estimator with 
	$P_{\rm max}=0$ and $M=1$ in comparison to MLSMC 
	estimator with different choices of $L$.}}
	\label{fig:anal}
\end{figure}

Finally, we consider the same last experiment except with the example of 
Sec. \ref{sec:example}. Again, over a wide range of MSE the unbiased estimator 
provides a significantly more efficient alternative to the MLSMC estimator. 
Here one can already observe the induced $P_{\rm max}=1$ bias 
for the unbiased estimator around $2^{-20}$.
Adjusting the various tuning parameters resulted in  
similar behaviour as was observed in the earlier 
experiments with the toy example. These results are not presented.

 \begin{figure}[!htbp]
	\centering\includegraphics[width=\textwidth]{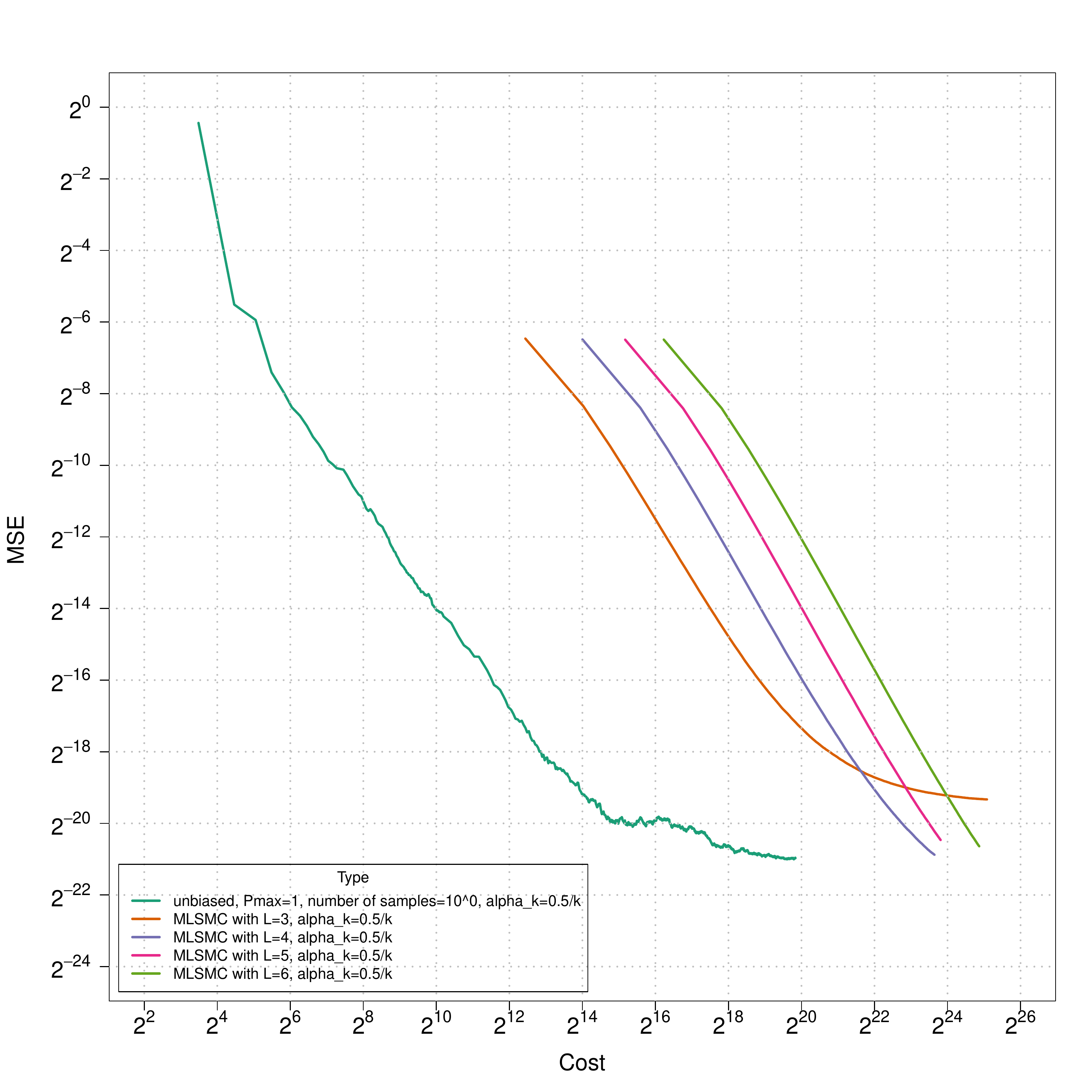}
	\caption{Example of Sec. \ref{sec:example}, SGD. 
	{MSE vs cost for $\alpha_k=0.5/k$, and unbiased estimator with 
	$P_{\rm max}=0$ and $M=1$ in comparison to MLSMC 
	estimator with different choices of $L$.}}
	\label{fig:nonanal}
\end{figure}

\subsubsection*{Acknowledgements}
AJ was supported by KAUST baseline funding.
KJHL was supported by The Alan Turing Institute under the 
EPSRC grant EP/N510129/1.

\appendix

\section{Proofs}\label{app:proofs}

We begin with some technical results that have been proved in previous works (e.g.~\cite{beskos,delm:04}).
Recall that $\mathbb{E}_{\theta,m}^N$ is an expectation w.r.t.~the probability with finite dimensional law \eqref{eq:mlsmc_law}, associated to the simulation
of Algorithm \ref{alg:mlsmc}.

\begin{lem}\label{lem:ml_bounds}
Assume (A\ref{hyp:A}-\ref{hyp:B}). Then for any $\theta\in\Theta$ there exists a $C<+\infty$ such that for any
$(l,N,\varphi)\in\mathbb{N}_0\times\mathbb{N}\times\mathcal{B}_b(\mathsf{X})$:
\begin{eqnarray*}
\mathbb{E}_{\theta,m}^N[[\eta_{\theta}^{l,N}-\eta_{\theta}^{l}](\varphi)^2] & \leq & \frac{C\|\varphi\|_{\infty}^2}{N} \\
|\mathbb{E}_{\theta,m}^N[[\eta_{\theta}^{l,N}-\eta_{\theta}^{l}](\varphi)]| & \leq & \frac{C\|\varphi\|_{\infty}}{N}.
\end{eqnarray*}
\end{lem}

\begin{proof}
The first statement is \cite[Theorem 7.4.4.]{delm:04} and the second follows easily from (e.g.) \cite[eq.~(A.2.), Lemma A.1.(iii)]{beskos}.
\end{proof}

Recall that we use $\mathbb{E}_{\theta}$ to denote expectation associated to the probability $\mathbb{P}_{\theta}$ in \eqref{eq:alg_2_law}
of which is associated to the generation of Algorithm \ref{alg:est_const}.

\begin{lem}\label{lem:est_pool}
Assume (A\ref{hyp:A}-\ref{hyp:B}). Then for any $\theta\in\Theta$ there exists a $C<+\infty$ such that for any
$(l,p,\varphi)\in\mathbb{N}_0\times\mathbb{N}_0\times\mathcal{B}_b(\mathsf{X})$, $1\leq N_0<N_1<\cdots<N_p<+\infty$:
$$
\mathbb{E}_{\theta}[[\eta_\theta^{l,N_{0:p}}-\eta_\theta^{l}](\varphi)^2] \leq \frac{C\|\varphi\|_{\infty}^2}{N_p}\Big(1+\frac{p^2}{N_p}\Big).
$$
\end{lem}

\begin{proof}
Follows by a similar approach to the proof of \cite[Proposition A.1.]{ubpf}, which needs the results in Lemma \ref{lem:ml_bounds}.
\end{proof}

\begin{lem}\label{lem:tech_lem1}
Assume (A\ref{hyp:A}-\ref{hyp:B}). Then for any $\theta\in\Theta$ there exists a $C<+\infty$ such that for any
$(l,p,i)\in\mathbb{N}_0\times\mathbb{N}_0\times\{1,\dots,d_{\theta}\}$, $1\leq N_0<N_1<\cdots<N_p<+\infty$:
$$
\mathbb{E}_{\theta}\Big[\Big(\frac{ \eta_\theta^{l,N_{0:p}}(G_{\theta}^l\{(\varphi_\theta^{l+1})_i-(\varphi_\theta^{l})_i\}) }
{\eta_\theta^{l,N_{0:p}}(G_{\theta}^l)}- 
\frac{ \eta_\theta^{l}(G_{\theta}^l\{(\varphi_\theta^{l+1})_i-(\varphi_\theta^{l})_i\}) }
{\eta_\theta^{l}(G_{\theta}^l)}
\Big)^2\Big] \leq \frac{C\|(\varphi_\theta^{l+1})_i-(\varphi_\theta^{l})_i\|_{\infty}^2}{N_p}\Big(1+\frac{p^2}{N_p}\Big).
$$
\end{lem}

\begin{proof}
As
$$
\frac{ \eta_\theta^{l,N_{0:p}}(G_{\theta}^l\{(\varphi_\theta^{l+1})_i-(\varphi_\theta^{l})_i\}) }
{\eta_\theta^{l,N_{0:p}}(G_{\theta}^l)}- 
\frac{ \eta_\theta^{l}(G_{\theta}^l\{(\varphi_\theta^{l+1})_i-(\varphi_\theta^{l})_i\}) }
{\eta_\theta^{l}(G_{\theta}^l)}
= 
$$
$$
\frac{1}{\eta_\theta^{l,N_{0:p}}(G_{\theta}^l)}\Big(
[\eta_\theta^{l,N_{0:p}}-\eta_\theta^{l}](G_{\theta}^l\{(\varphi_\theta^{l+1})_i-(\varphi_\theta^{l})_i\})  
\Big) + \frac{\eta_\theta^{l}(G_{\theta}^l\{(\varphi_\theta^{l+1})_i-(\varphi_\theta^{l})_i\})}
{\eta_\theta^{l,N_{0:p}}(G_{\theta}^l)\eta_\theta^{l}(G_{\theta}^l)}[\eta_\theta^{l}-\eta_\theta^{l,N_{0:p}}](G_{\theta}^l)
$$
one can simply use the $C_2$-inequality, (A\ref{hyp:B}) and Lemma \ref{lem:est_pool}  to complete the proof.
\end{proof}

\begin{lem}\label{lem:tech_lem2}
Assume (A\ref{hyp:A}-\ref{hyp:C}). Then for any $\theta\in\Theta$ there exists a $C<+\infty$ such that for any
$(l,p,i)\in\mathbb{N}\times\mathbb{N}_0\times\{1,\dots,d_{\theta}\}$, $1\leq N_0<N_1<\cdots<N_p<+\infty$:
$$
\mathbb{E}_{\theta}\Bigg[\Bigg(\frac{ \eta_\theta^{l-1,N_{0:p}}(G_{\theta}^{l-1}(\varphi_\theta^{l})_i) }
{\eta_\theta^{l-1,N_{0:p}}(G_{\theta}^{l-1})}- \eta_\theta^{l-1,N_{0:p}}((\varphi_\theta^{l-1})_i) -
\Big(
\frac{ \eta_\theta^{l-1}(G_{\theta}^{l-1}(\varphi_\theta^{l})_i) }
{\eta_\theta^{l-1}(G_{\theta}^{l-1})}- \eta_\theta^{l-1}((\varphi_\theta^{l-1})_i)
\Big)
\Bigg)^2\Bigg] \leq 
$$
$$
\frac{C}{N_p}\Big(1+\frac{p^2}{N_p}\Big)\Big(\|(\varphi_\theta^{l})_i-(\varphi_\theta^{l-1})_i\|_{\infty}^2+\Big\|G_{\theta}^{l-1}\frac{Z_\theta^{l-1}}{Z_\theta^{l}}-1\Big\|_{\infty}^2\Big).
$$
\end{lem}

\begin{proof}
We have the decomposition
$$
\frac{ \eta_\theta^{l-1,N_{0:p}}(G_{\theta}^{l-1}(\varphi_\theta^{l})_i) }
{\eta_\theta^{l-1,N_{0:p}}(G_{\theta}^{l-1})}- \eta_\theta^{l-1,N_{0:p}}((\varphi_\theta^{l-1})_i) -
\Big(
\frac{ \eta_\theta^{l-1}(G_{\theta}^{l-1}(\varphi_\theta^{l})_i) }
{\eta_\theta^{l-1}(G_{\theta}^{l-1})}- \eta_\theta^{l-1}((\varphi_\theta^{l-1})_i)
\Big) = \sum_{j=1}^3 T_j
$$
where
\begin{eqnarray*}
T_1 & = & \frac{ \eta_\theta^{l-1,N_{0:p}}(G_{\theta}^{l-1}\{(\varphi_\theta^{l})_i-(\varphi_\theta^{l-1})_i\}) }
{\eta_\theta^{l-1,N_{0:p}}(G_{\theta}^{l-1})} - 
\frac{ \eta_\theta^{l-1}(G_{\theta}^{l-1}\{(\varphi_\theta^{l})_i-(\varphi_\theta^{l-1})_i\}) }
{\eta_\theta^{l-1}(G_{\theta}^{l-1})} \\
T_2 & = & -\frac{\eta_\theta^{l-1,N_{0:p}}(G_{\theta}^{l-1}(\varphi_\theta^{l-1})_i)}{\eta_\theta^{l-1,N_{0:p}}(G_{\theta}^{l-1})}
\eta_\theta^{l-1,N_{0:p}}\Big(G_{\theta}^{l-1}\frac{Z_\theta^{l-1}}{Z_\theta^{l}}-1\Big) +
\frac{\eta_\theta^{l-1}(G_{\theta}^{l-1}(\varphi_\theta^{l-1})_i)}{\eta_\theta^{l-1}(G_{\theta}^{l-1})}
\eta_\theta^{l-1}\Big(G_{\theta}^{l-1}\frac{Z_\theta^{l-1}}{Z_\theta^{l}}-1\Big)\\
T_3 & = & \eta_\theta^{l-1,N_{0:p}}\Big((\varphi_\theta^{l-1})_i\Big(G_{\theta}^{l-1}\frac{Z_\theta^{l-1}}{Z_\theta^{l}}-1\Big)\Big) - 
\eta_\theta^{l-1}\Big((\varphi_\theta^{l-1})_i\Big(G_{\theta}^{l-1}\frac{Z_\theta^{l-1}}{Z_\theta^{l}}-1\Big)\Big).
\end{eqnarray*}
Thus, one can apply the $C_2-$inequality twice and deal individually with the terms $\sum_{j=1}^3 \mathbb{E}_{\theta}[T_j^2]$.
For $\mathbb{E}_{\theta}[T_1^2]$ one can use Lemma \ref{lem:tech_lem1}. For $\mathbb{E}_{\theta}[T_3^2]$ one can use Lemma \ref{lem:est_pool}.
So to conclude, we consider $\mathbb{E}_{\theta}[T_2^2]$. We have
$$
T_2  =  T_4 + T_5
$$
where
\begin{eqnarray*}
T_4 & = &  
-\Big(\frac{\eta_\theta^{l-1,N_{0:p}}(G_{\theta}^{l-1}(\varphi_\theta^{l-1})_i)}{\eta_\theta^{l-1,N_{0:p}}(G_{\theta}^{l-1})} - 
\frac{\eta_\theta^{l-1}(G_{\theta}^{l-1}(\varphi_\theta^{l-1})_i)}{\eta_\theta^{l-1}(G_{\theta}^{l-1})}\Big)\eta_\theta^{l-1,N_{0:p}}\Big(G_{\theta}^{l-1}\frac{Z_\theta^{l-1}}{Z_\theta^{l}}-1\Big)
\\
T_5 & = &  \frac{\eta_\theta^{l-1}(G_{\theta}^{l-1}(\varphi_\theta^{l-1})_i)}{\eta_\theta^{l-1}(G_{\theta}^{l-1})}\Big(
\eta_\theta^{l-1}\Big(G_{\theta}^{l-1}\frac{Z_\theta^{l-1}}{Z_\theta^{l}}-1\Big) - 
\eta_\theta^{l-1,N_{0:p}}\Big(G_{\theta}^{l-1}\frac{Z_\theta^{l-1}}{Z_\theta^{l}}-1\Big)
\Big).
\end{eqnarray*}
Applying the $C_2-$inequality once again allows one to consider just $\mathbb{E}_{\theta}[T_4^2]$ and $\mathbb{E}_{\theta}[T_5^2]$ individually.
For $\mathbb{E}_{\theta}[T_5^2]$ one can use (A\ref{hyp:C}) and Lemma \ref{lem:est_pool}. As
$$
T_4 = -
\Big(\frac{[\eta_\theta^{l-1,N_{0:p}}-\eta_\theta^{l-1}](G_{\theta}^{l-1}(\varphi_\theta^{l-1})_i)}{\eta_\theta^{l-1,N_{0:p}}(G_{\theta}^{l-1})} - 
\frac{\eta_\theta^{l-1}(G_{\theta}^{l-1}(\varphi_\theta^{l-1})_i)}{\eta_\theta^{l-1}(G_{\theta}^{l-1})\eta_\theta^{l-1,N_{0:p}}(G_{\theta}^{l-1})}
[\eta_\theta^{l-1}-\eta_\theta^{l-1,N_{0:p}}](G_{\theta}^{l-1})
\Big)\eta_\theta^{l-1,N_{0:p}}\Big(G_{\theta}^{l-1}\frac{Z_\theta^{l-1}}{Z_\theta^{l}}-1\Big).
$$
One can then conclude the result by applying the $C_2-$inequality and using (A\ref{hyp:B}) and Lemma \ref{lem:est_pool}.
\end{proof}


\begin{thebibliography}{99}

\bibitem{sergios}
{\sc Agapiou}, S., {\sc Roberts}, G. O.~\& {\sc Vollmer}, S.~(2018). Unbiased Monte Carlo: 
Posterior estimation for intractable/infinite-dimensional models.
\emph{Bernoulli}, {\bf 24}, 1726--1786.

\bibitem{beskos}
{\sc Beskos}, A., 
{\sc Jasra}, A., 
{\sc Law}, K. J. H., 
{\sc Tempone}, R., \& 
{\sc Zhou}, Y.~ (2017). 
Multilevel Sequential Monte Carlo samplers. {\em Stoch. Proc. Appl.}, {\bf 127}, 1417-1440.



\bibitem{brenner}
{\sc Brenner}, S., \& 
{\sc Scott}, R. (2007). 
\textit{The Mathematical Theory of Finite Element Methods}. Springer: New York.



\bibitem{ciarlet}
{\sc Ciarlet}, P. G. (2002). 
\textit{The Finite Element Method for Elliptic Problems}.  SIAM: Philadelphia.


\bibitem{delm:13}
{\sc Del Moral}, P.~(2013). \textit{Mean Field Simulation for Monte Carlo Integration}. Chapman \& Hall: London.


\bibitem{delm:04}
{\sc Del Moral}, P.~(2004). \textit{Feynman-Kac Formulae: Genealogical and
Interacting Particle Systems with Applications}. Springer: New York.

\bibitem{engl}
{\sc Engl, H. W., Hanke, M.}, \& {\sc Neubauer, A.} (1996). 
\textit{Regularization of inverse problems}.  Springer: New York.

\bibitem{franklin}
{\sc Franklin}, J. N. (1970). 
Well-posed stochastic extensions of ill-posed linear problems. 
\textit{Journal Math. Anal. App.}, 
{\bf 31}(3), 682-716.

\bibitem{ubpf}
{\sc Jasra}, A., {\sc Law}, K. J. H., \& {\sc Yu}, F.~(2020). Unbiased filtering of a class of partially observed diffusions. arXiv preprint.



\bibitem{mcl}
{\sc McLeish}, D.~(2011). A general method for debiasing a Monte Carlo estimator. \emph{Monte Carlo Meth. Appl.}, {\bf 17}, 301--315.

\bibitem{rhee}
{\sc Rhee}, C. H. \& {\sc Glynn}, P.~(2015). Unbiased estimation with square root convergence for SDE models. \emph{Op. Res.}~{\bf 63}, 1026--1043. 

\bibitem{stuart}
{\sc Stuart}, A. M. (2010). 
Inverse problems: A Bayesian perspective. 
\textit{Acta Numerica}, {\bf 19}, 451-559.

\bibitem{tadic}
{\sc Tadic}, V. \& {\sc Doucet}, A.~(2018). Asymptotic properties of recursive maximum likelihood estimators in non-linear state-space models.
arXiv preprint.

\bibitem{tikhonov}
{\sc Tikhonov}, A. N., \& {\sc Glasko}, V. B. (1964). 
The approximate solution of Fredholm integral equations of the first kind. 
\textit{USSR Comp. Math. Math. Phys.}, 
{\bf 4}(3), 236-247.

\bibitem{vihola}
{\sc Vihola}, M.~(2018). Unbiased estimators and multilevel Monte Carlo. \emph{Op. Res.}, {\bf 66}, 448--462.

\end{thebibliography}
\end{document}